\documentclass[acmsmall]{acmart}

\AtBeginDocument{%
  }

\setcopyright{cc}
\setcctype{by}
\acmJournal{PACMPL}
\acmYear{2025} 
\acmVolume{9} 
\acmNumber{POPL} 
\acmArticle{3} 
\acmMonth{1} 
\acmPrice{}
\acmDOI{10.1145/3704839}
\received{2024-07-11}
\received[accepted]{2024-11-07}

\usepackage{algorithm2e}
\usepackage{wrapfig}
\usepackage{subcaption}
\usepackage{amsmath}
\usepackage{accents}
\usepackage{listings}

\lstdefinestyle{CStyle}{
    basicstyle=\footnotesize,
    keywordstyle=\bfseries,
    morekeywords={assert, for, if, is, to, S0, S1, S2, S3, S4, C0, C1, C2, C3, C4, C5, inputs, outputs, locals, let, case, and, or, affine, external, reduce, max, min},
    moredelim=**[is][\color{red}]{@r@}{@r@},
    moredelim=**[is][\color{blue}]{@b@}{@b@},
    moredelim=**[is][\color{darkgreen}]{@g@}{@g@},
    moredelim=**[is][\color{gray}]{@k@}{@k@},
    moredelim=**[is][\color{white}]{@w@}{@w@},
    frame=single
}

\usepackage[colorinlistoftodos]{todonotes} 

\newcommand{\halfvec}[1]{\accentset{\rightharpoonup}{#1}}

\RestyleAlgo{ruled}

\definecolor{darkgreen}{RGB}{0,153,0}

\begin{document}

\title{Maximal Simplification of Polyhedral Reductions}

\author{Louis Narmour}
\orcid{0009-0009-3298-5282}
\affiliation{%
  \institution{Colorado State University}
  \city{Fort Collins}
  \country{USA}
}
\affiliation{%
  \institution{University of Rennes, Inria, CNRS, IRISA}
  \city{Rennes}
  \country{France}
}
\email{louis.narmour@colostate.edu}

\author{Tomofumi Yuki}
\orcid{0000-0002-5737-6178}
\affiliation{%
  \institution{Unaffiliated}
  \city{Yokohama}
  \country{Japan}
}
\email{tomofumi.yuki@gmail.com}

\author{Sanjay Rajopadhye}
\orcid{0000-0002-4246-6066}
\affiliation{%
  \institution{Colorado State University}
  \city{Fort Collins}
  \country{USA}
}
\email{sanjay.rajopadhye@colostate.edu}





\begin{abstract}
\emph{Reductions} combine collections of input values with an associative and often commutative operator to produce collections of results.
When the \emph{same} input value contributes to \emph{multiple} outputs, there is an opportunity to \emph{reuse} partial results, enabling \emph{reduction simplification}.
Simplification often produces a program with lower asymptotic complexity.
Typical compiler optimizations yield, at best, a constant fold speedup, but a complexity improvement from, say, cubic to quadratic complexity yields unbounded speedup for sufficiently large problems.
It is well known that reductions in polyhedral programs may be simplified \emph{automatically}, but previous methods cannot exploit all available reuse.  
This paper resolves this long-standing open problem, thereby attaining minimal asymptotic complexity in the simplified program.
We propose extensions to prior work on simplification to support any independent commutative reduction.
At the heart of our approach is piece-wise simplification, the notion that we can split an arbitrary reduction into pieces and then independently simplify each piece.
However, the difficulty of using such piece-wise transformations is that they typically involve an infinite number of choices.
We give constructive proofs to deal with this and select a finite number of pieces for simplification.

\end{abstract}




\maketitle

\section{Introduction}

Computing technology has become increasingly powerful and complex over the years, offering more capability with each new generation of processors.
For example, the latest generation of Intel Xeon processors supports configurations of up to 50 cores on a single die with 100 MB of last-level cache.
However, using all the available processing power in the presence of complex data dependencies is not always easy.
Relying on traditional compilers to produce high-performance code for a given input program is often insufficient.
One reason is that general-purpose language compilers must be very conservative in the types of optimizations that can be employed.
Consequently, the onus is on the application developer to write the program in such a way that the compiler can successfully detect optimization opportunities.
This is challenging because it is often easier to think about problems from a higher level of abstraction, whereas writing efficient code requires lower-level reasoning.
Over the years, this has led to the development of a wide variety of Domain Specific Languages (DSLs) and highly specialized frameworks.

The polyhedral model~\cite{rajopadhye_synthesizing_1989, feautrier_dataflow_1991, feautrier_efficient_1992, amarasinghe_communication_1993, fortes_data_1984, irigoin_supernode_1988, lam_systolic_1989, lengauer_loop_1993, pugh_omega_1991, quinton_mapping_1989, ramanujam_nested_1990, schreiber_automatic_1990, wolf_loop_1991, wolf_data_1991, wolfe_iteration_1987} is one such framework and is a mathematical formalism for specifying, analyzing, and transforming compute- and data-intensive programs.
Such programs occur in a wide variety of application domains, like dense linear algebra, signal and image processing, convolutional neural nets, deep learning, back-propagation training, and dynamic programming, to name just a few.
In this paper, we study the optimization of programs that can be specified by reductions within the polyhedral model.
Reductions are ubiquitous in computing and typically involve applying an associative, often commutative, operator to collections of inputs to produce one or more results.
Such operations are interesting because they often require special handling to obtain good performance.
The OpenMP C/C++ multithreading API~\cite{openmp_architecture_review_board_openmp_2021} even has directives tailored specifically for parallelizing 
reductions.
However, in some cases, the input program specification may involve \textit{reuse}---the same value contributing to multiple results.
When properly exploited, it is possible to improve the asymptotic complexity of such programs.
A complexity improvement from, say, $O(N^3)$ to $O(N^2)$ yields unbounded speedup since, asymptotically, $N$ can be arbitrarily large.

\begin{wrapfigure}[18]{tr}{0.5\textwidth}
    \vspace{-3mm}
    \includegraphics[width=0.5\textwidth]{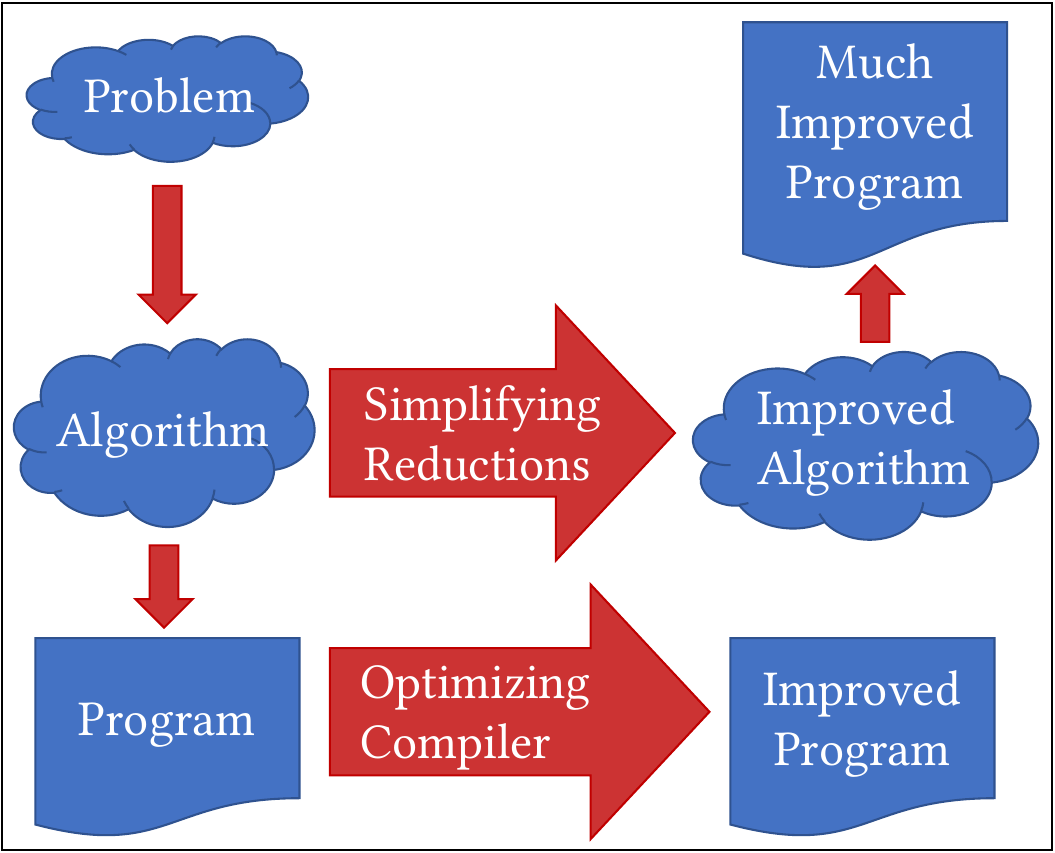}
    \caption{Simplification improves the asymptotic complexity of the \emph{algorithm}.}
    \label{fig:complexity}
\end{wrapfigure}

Gautam and Rajopadhye~\cite{gautam_simplifying_2006} previously showed how to reduce, by such polynomial degrees, the asymptotic complexity of commutative polyhedral reductions.
They developed a program transformation called \emph{simplification} and outlined a recursive algorithm to automatically simplify polyhedral reductions.  
Their work (henceforth referred to as GR06) is optimal within its scope of applicability.
The simplification algorithm may fail when the reduction operator does not admit an inverse.
Such operators are very common in many dynamic programming algorithms.
Indeed, polyadic dynamic programming problems~\cite{li_systolic_1985} are nothing but reductions and are widely used in many bio-informatics algorithms~\cite{lyngso_fast_1999, nussinov_fast_1980, zuker_optimal_1981, wolfinger_efficient_2004, flamm_rna_2000, zadeh_nupack_2011, chitsaz_birna_2009,lorenz_viennarna_2011, mathews_prediction_2006, boniecki_simrna_2016, bringmann_truly_2019, lorenz_rna_2016, huang_linearfold_2019}.

This paper proposes a method to extend the simplification algorithm to handle all scenarios, achieving maximal simplification.
At the heart of our approach is \textit{piece-wise simplification}, the notion that we can split the problem into a number of pieces and then independently simplify each piece.
Such \emph{piece-wise affine transformations}, also known as index-set splitting, have a long history in other polyhedral analyses~\cite{rajopadhye_piecewise_1992, griebl_index_2000, bondhugula_tiling_2014, vasilache_automatic_2007, razanajato_splitting_2017}.
The difficulty lies in the fact that, in general, there are infinitely many ways to split a polyhedron. 
We give constructive proofs showing how to select a finite number of pieces for simplification.
In doing so, we make the following contributions,
\begin{enumerate}
    \item We propose extensions to the GR06 simplification algorithm to support any arbitrary independent reduction, particularly when the operator does not admit an inverse.
    \item We provide an implementation of our approach as an accompanying software artifact~\cite{narmour_maximal_2024}.
\end{enumerate}

The remainder of this paper is organized as follows.
Section~\ref{sec:motivating-examples} provides several motivating examples of progressively increasing difficulty.
In Section~\ref{sec:background}, we review background material on the polyhedral model and prior work to explain the limitations of simplification.
In Sections~\ref{sec:extending-simplification} and~\ref{sec:problem-formulation}, we formulate our approach and justify it in Sections~\ref{sec:dD-reductions} and~\ref{sec:2D-reductions (triangles)}.
We discuss aspects of our implementation in Section~\ref{sec:implementation}.
Finally, we review related work in Section~\ref{sec:related-work} and conclude in Section~\ref{sec:conclusion}.

\section{Motivating Examples} \label{sec:motivating-examples}

We now show several simple examples to illustrate the intuition of simplification and its limitations.

\subsection{Prefix Sum} \label{sec:motiv-prefix-sum}

Consider the prefix sum, which computes an N-element output array from an N-element input array where the $i$'th element of the output is the sum of the first $i$ elements of the input:
\begin{equation} \label{eq:motiv-prefix-sum}
    Y_{i} = \sum_{j=0}^{j \leq i} X_{j}
\end{equation}
As written, Equation~\ref{eq:motiv-prefix-sum} has an asymptotic complexity of $O(N^2)$ because each of the N elements in the output involves a summation over $O(N)$ input values.
However, two adjacent elements in the output involve summing many of the \emph{same elements} (i.e., there is \emph{reuse}), which means there is an opportunity to reuse partial results.

\begin{figure}[h]
    \centering
    \begin{subfigure}{0.4\textwidth}
        \centering
        \includegraphics[width=\textwidth]{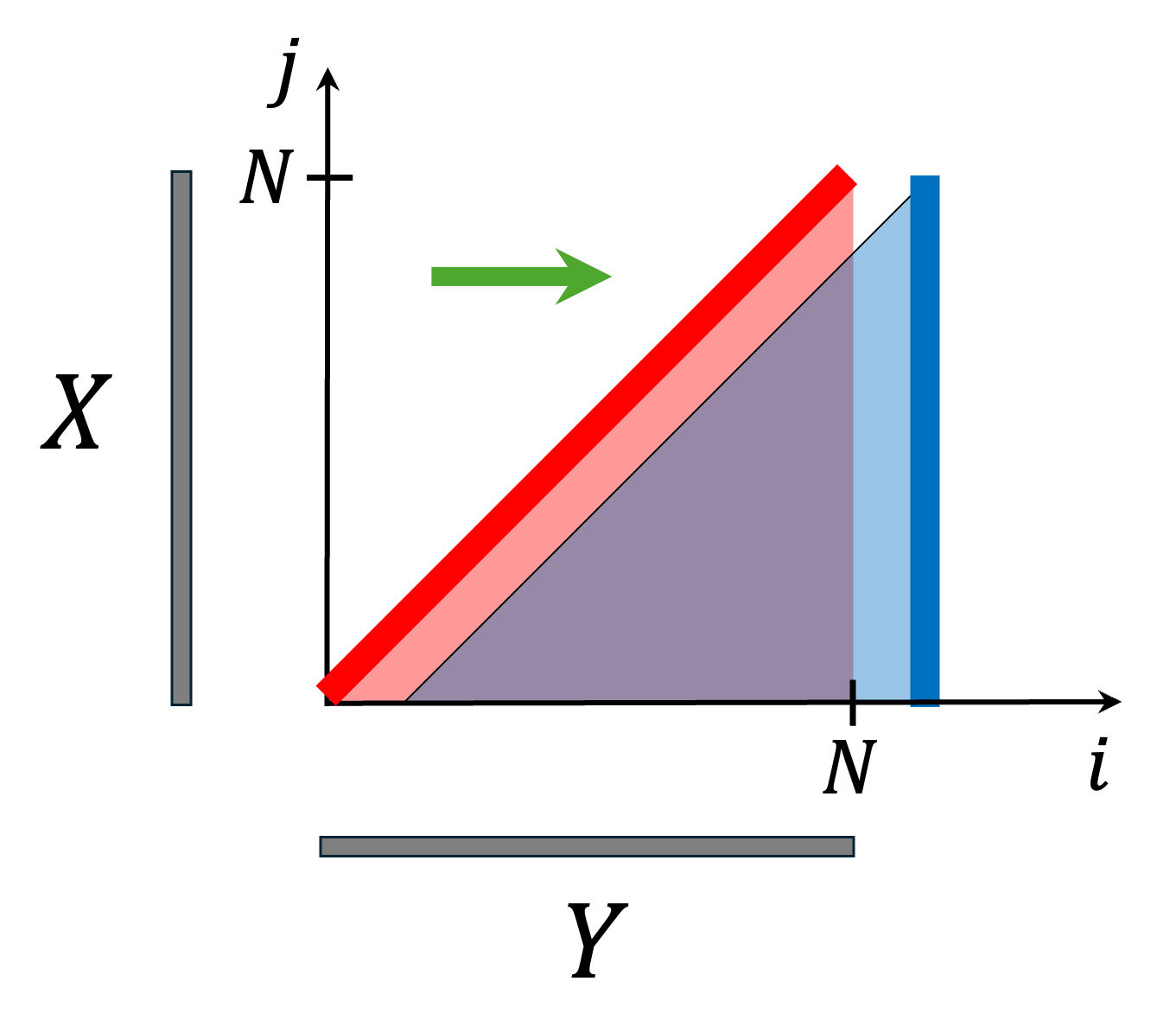}
        \caption{Equation~\ref{eq:motiv-prefix-sum-simplified-v1} expressing $Y_{i}$ in terms of $Y_{i-1}$.}
    \end{subfigure}%
    \hspace{5mm}
    \begin{subfigure}{0.4\textwidth}
        \centering
        \includegraphics[width=\textwidth]{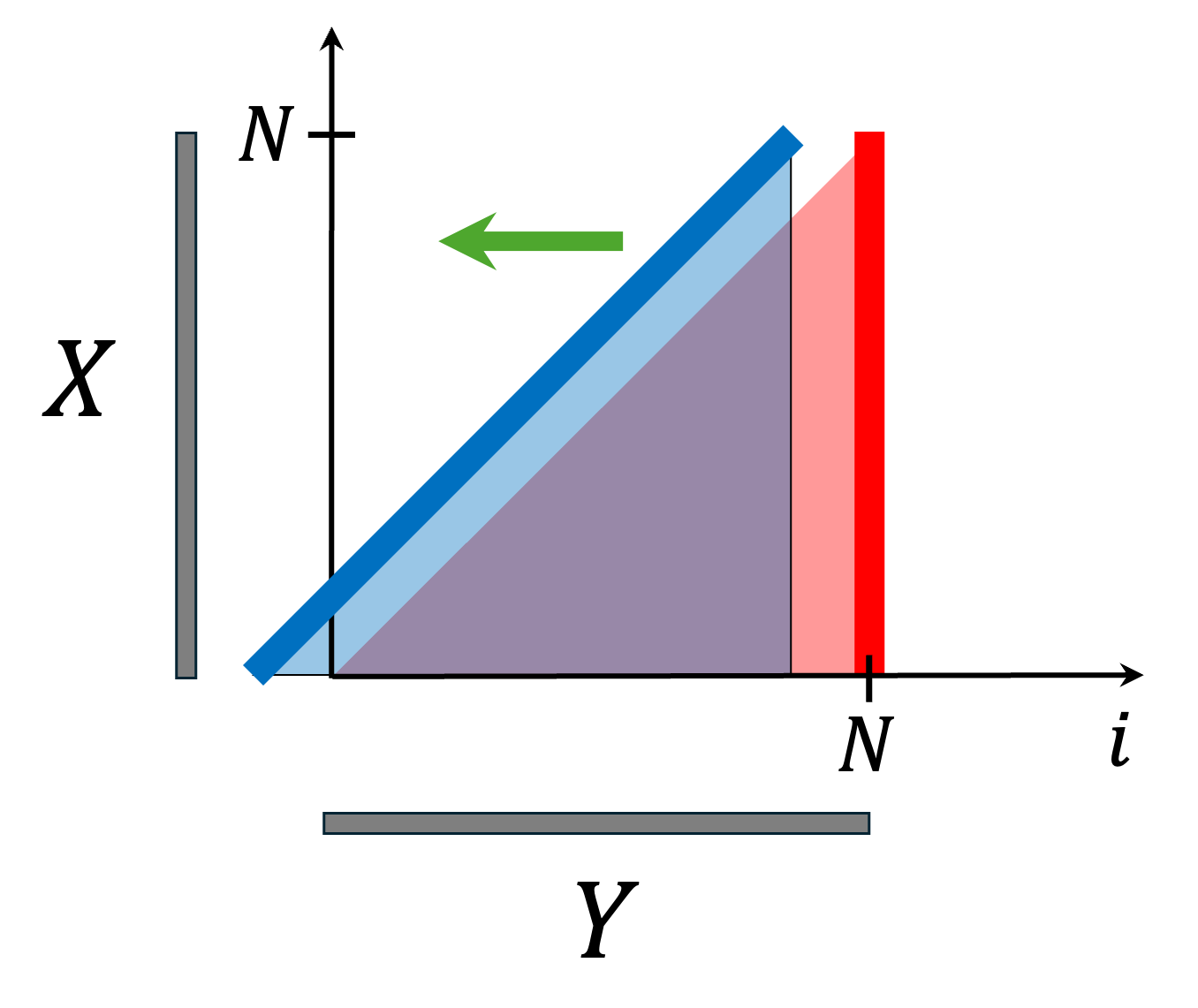}
        \caption{Equation~\ref{eq:motiv-prefix-sum-simplified-v2} expressing $Y_{i}$ in terms of $Y_{i+1}$.}
    \end{subfigure}
    \caption{Reduction from quadratic to linear asymptotic complexity. The input variable $X$ is oriented along the vertical axis, and the output $Y$ is oriented along the horizontal.}
    
    \Description{Prefix sum}
    \label{fig:prefix-sum}
\end{figure}

For example, $Y_{k} = X_{0} + X_{1} + ... + X_{k}$ and 
$Y_{k+1} = X_{0} + X_{1} + ... + X_{k} + X_{k+1}$.
Once $Y_{k}$ is computed, there is no need to re-sum the first $k$ elements of $X$ to compute $Y_{k+1}$.
Instead, $Y_{k+1}$ can be expressed in terms of $Y_{k}$ and only the \emph{new} values of $X$:
\begin{align*}
    Y_{k+1} &= (X_{0} + X_{1} + X_{2} + ... + X_{k}) + X_{k+1} \\
    Y_{k+1} &= Y_{k} + X_{k+1}
\end{align*}
This observation allows Equation~\ref{eq:motiv-prefix-sum} to be rewritten as either:
\begin{equation} \label{eq:motiv-prefix-sum-simplified-v1}
    Y_{i} = \left\{
        \begin{array}{lcl} 
            i=0 &:& \textcolor{red}{X_{0}} \\[1mm]
            i > 0 &:& \textcolor{darkgreen}{Y_{i-1}} + \textcolor{red}{X_{i}}
        \end{array} \right.
\end{equation}
where $Y_{i}$ is computed from $Y_{i-1}$, or equally validly, as: 
\begin{equation} \label{eq:motiv-prefix-sum-simplified-v2}
    Y_{i} = \left\{
        \begin{array}{lcl} 
            i=N &:&\displaystyle \textcolor{red}{\sum_{j=0}^{N} X_{j}} \\[2mm]
            i < N &:& \textcolor{darkgreen}{Y_{i+1}} - \textcolor{blue}{X_{i+1}}
        \end{array} \right.
\end{equation}
where $Y_{i}$ is computed from $Y_{i+1}$ instead.
The terms are colored to match the corresponding edges in Figure~\ref{fig:prefix-sum}.
Consequently, both have a smaller, better asymptotic complexity of $O(N)$.
The latter requires an initial summation of all N elements of $X$ to produce $Y_{N}$, but the overall complexity is still linear because the remaining $N-1$ elements are each computed in constant time.

The arrows in Figure~\ref{fig:prefix-sum} reflect the choice of expressing $Y_{i}$ as either $Y_{i-1}$ or $Y_{i+1}$, which has the effect of moving the computation from the 2D domain (triangle) to some of the 1D edges.
The correspondance between Figure~\ref{fig:prefix-sum} and Equations~\ref{eq:motiv-prefix-sum-simplified-v1} and~\ref{eq:motiv-prefix-sum-simplified-v2} will be made more explicit in Section~\ref{sec:background}. 
The goal of GR06's simplification is to explore all such possible rewrites and find the ones with the optimal asymptotic complexity in the presence of reuse.
However, as illustrated by the next example, some rewrites may not be possible.

\subsection{Prefix Max} \label{sec:motiv-prefix-max}

Consider the prefix max, which is identical to the previous example, except it uses the max operator instead of addition:
\begin{equation} \label{eq:motiv-prefix-max}
    Y_{i} = \max_{j=0}^{j \leq i} X_{j}
\end{equation}
Everything else is the same, the value produced at $Y_{i}$ can be used to compute the next value at $Y_{i+1}$, allowing it to be rewritten with $O(N)$ complexity as:
\begin{equation} \label{eq:motiv-prefix-max-simplified-v1}
    Y_{i} = \left\{
        \begin{array}{lcl} 
            i=0 &:& X_{0} \\[1mm]
            i > 0 &:& \mathrm{max}(Y_{i-1}, X_{i})
        \end{array} \right.
\end{equation}
exactly like Equation~\ref{eq:motiv-prefix-sum-simplified-v1}.
However, the key difference here is that only one of the rewrites is possible.
There is no way to express an equation analogous to Equation~\ref{eq:motiv-prefix-sum-simplified-v2} because the max operator does not admit an inverse.
This is not an issue here since at least one rewrite does not involve the inverse operation, and thus GR06 can find it.
However, in general, a rewrite may not exist that does not involve the inverse operation, and consequently, GR06 may fail, as illustrated by the next example.

\subsection{Sliding and Increasing Max Filter} \label{sec:motiv-i-to-2i}

Consider the following equation, which computes an $N$-element output array from an $N$-element input array where the $i$'th element of the output is the max over the sliding, and increasing, window from $i$ to $2i$ on the input: 
\begin{equation} \label{eq:motiv-i-to-2i}
    Y_{i} = \max_{j=i}^{j \leq 2i} X_{j}
\end{equation}
As written, this has an asymptotic complexity of $O(N^2)$.
Like the previous examples, there is also reuse.
Across two adjacent answers in the output, \textcolor{darkgreen}{$O(N)$ of the same inputs} are read:
\begin{align*}
    Y_{k} &= \mathrm{max}(\textcolor{blue}{X_{k}}, \textcolor{darkgreen}{X_{k+1}, X_{k+2},  ..., X_{2k}}) \\
    Y_{k+1} &= \mathrm{max}(\hspace{5mm} \textcolor{darkgreen}{X_{k+1}, X_{k+2},  ..., X_{2k}}, \textcolor{red}{X_{2k+1}, X_{2k+2}}) \\
\end{align*}
But, neither can $Y_{k+1}$ be expressed in terms of $Y_{k}$ nor can $Y_{k}$ be expressed in terms of $Y_{k+1}$ because this would require \emph{removing} the contributions of either \textcolor{blue}{$X_{k}$} or \textcolor{red}{$X_{2k+1}$} and \textcolor{red}{$X_{2k+2}$} which is not possible because the max operator has no inverse.

However, it is still possible to rewrite Equation~\ref{eq:motiv-i-to-2i} with an asymptotic complexity of $O(N)$.
The solution is to split this reduction into three pieces by the hyperplanes $2i=N$ and $j=N$.
Since the reduction is commutative, the order of accumulation does not matter, and as we discuss in Section~\ref{sec:split-reduction}, Equation~\ref{eq:motiv-i-to-2i} may be rewritten as,
\begin{align}
    Y_{i} = \left\{
        \begin{array}{lcl} 
            2i \geq N &:& \displaystyle \max \bigg( \Big(\max_{j=i}^{j < N} X_{j}\Big), \Big( \max_{j=N}^{j \leq 2i} X_{j} \Big) \bigg)  \\[1mm]
            2i<N &:&\displaystyle \max_{j=i}^{j \leq 2i} X_{j}
        \end{array} \right. 
\end{align}
In this form, we can see three pieces.
The two reductions in the top branch $2i \geq N$ are both instances of a standard suffix max and prefix max in the previous section,~\ref{sec:motiv-prefix-max}.
The reduction in the second branch is the same reduction as Equation~\ref{eq:motiv-i-to-2i} but just over a smaller domain.
This leads to a recursive simplification strategy, which can be further and similarly split into the same number of pieces.
We formulate this more precisely in Section~\ref{sec:2D-reductions (triangles)}, and the rest of this work is dedicated to generalizing this idea to arbitrary input reductions.

\section{Background} \label{sec:background}

In this section, we summarize the simplification transformation of GR06 with an example.

\subsection{Terminology and Notation} 

Generally, simplification operates on computations of the following form:
\begin{eqnarray}
    Y_{f_{p}(z)}
    &=& \bigoplus_{z \in \mathcal{D}} X_{f_{d}(z)} \label{eq:reduction-def}
\end{eqnarray}
where $Y$ is an output variable, and $f_{p}$ and $f_{d}$ are affine functions.  
The \emph{body} is an arbitrary expression involving other variables in the program, provided they do not depend recursively, even transitively, on any instance of $Y$.
Without loss of generality (our tools handle arbitrary expressions), we abstract it as an input variable $X$.  We use the terminology of GR06, summarized below:

\begin{itemize}
    \item \textit{Polyhedron}: A set of integer points defined by a finite list of inequality and equality constraints.
    \item \textit{Reduction body} ($\mathcal{D}$): A $d$-dimensional polyhedron representing the values of the program variable indices involved in the reduction's accumulation.
    \item \textit{Facet (or face)}: A $k$-dimensional face of the reduction body described uniquely by a subset of its inequality constraints treated as equalities.
    \item \textit{Face lattice}: The hierarchical arrangement of faces of the reduction body, see Section~\ref{sec:background-face-lattice}.
    \item \textit{Write function} ($f_p$): A rank-deficient affine map from $\mathbb{Z}^{d} \rightarrow \mathbb{Z}^{d-a}$ defining to which element of the output each point in the reduction body accumulates.\footnote{Rank-deficiency implies that $a > 0$.}
    \item \textit{Accumulation space} ($\mathcal{A}$): The $a$-dimensional space characterized by the null space of the write function.
    \item \textit{Read function} ($f_d$): A (potentially rank-deficient) affine map from $\mathbb{Z}^{d} \rightarrow \mathbb{Z}^{d-r}$ characterizing from which element of the input each point in the reduction body reads.
    \item \textit{Reuse space} ($\mathcal{R}$): The $r$-dimensional space characterized by the null space of the read function.
    \item \textit{Reuse vector ($\halfvec{\rho}$)}: Any vector in the reuse space.
\end{itemize}

Such programs consist of arbitrarily nested loops with affine control and a single statement in the body, plus initialization statements as appropriate.
Each statement accumulates values into some element of an output array using the reduction operator $\oplus$.
The right-hand sides of the statements are arbitrary expressions evaluated in constant time, reading other array variables via affine access functions.
Our reduction operators are associative and commutative, so the execution order of the loops is irrelevant.
Every instance of the loop body (for every legal value of the surrounding indices) can be treated as a new dummy variable.
Simplification is possible when the reuse space is non-empty (i.e., the read function, $f_{d}$, is rank-deficient).

\subsection{Face Lattice} \label{sec:background-face-lattice}

The face lattice~\cite{loechner_parameterized_1997} is an important data structure for simplification.
The face lattice of a polyhedron $\mathcal{D}$  is a graph whose nodes are the \emph{facets} of $\mathcal{D}$.
Each face in the lattice is the intersection of $\mathcal D$ with one or more \emph{equalities} of the form $\alpha z+\gamma = 0$ for $z \in \mathcal{D}$ obtained by \emph{saturating} one or more of the inequality constraints in $\mathcal{D}$.
We refer to each $k$-dimensional face as a \emph{($k$)-face}.
More than one constraint may be saturated to yield recursively, facets of facets, or \emph{faces}.
In the lattice, faces are arranged level by level, and each face saturates exactly one constraint in addition to those saturated by its immediate ancestors.

\begin{figure}[tbh]
    \centering
    \includegraphics[width=0.8\textwidth]{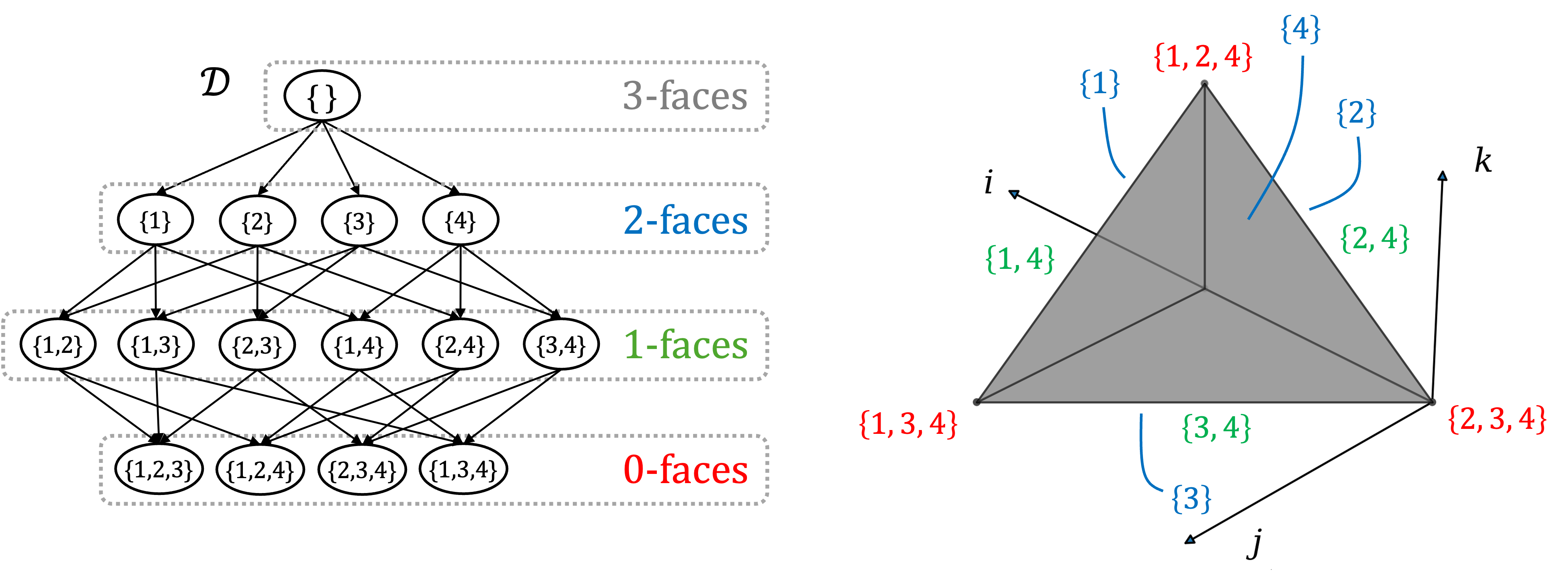}
    \caption{Face lattice of $\mathcal{D}$ in Equation~\ref{eq:sr-example}. The numbers inside the labels denote the constraints that, when saturated, describe the face. For example, ``$\{4\}$'' denotes saturating constraint $c_{4} : k \leq i-j$, which represents the top oblique 2-face. The back three edges (1-faces) and vertex (0-face) are not labeled.}
    \Description{face lattice}
    \label{fig:face-lattice}
\end{figure}



\subsection{Working Example} \label{sec:background-working-ex}

Consider the reduction which produces an $O(N)$-element array specified by the following equation: 
\begin{equation}
  \label{eq:sr-example}
   Y_{i} = \sum_{(i,j,k) \in \mathcal{D}} X_{k}
\end{equation}
over the 3-dimensional domain:
\begin{equation}
    \mathcal{D} = \{ [i,j,k] \mid (i <= N) \land (0 \leq j) \land (0 \leq k) \land (k \leq i-j) \}
\end{equation}
where $i$, $j$, and $k$ are indices and $N$ is a parametric size parameter.
The write function here is $f_{p} = \{[i,j,k] \rightarrow [i] \}$ and therefore the accumulation space is the 2-dimensional $jk$-plane, $\mathcal{A} = \{[i,j,k] \mid i=0\}$.
Similarly, the read function here is $f_{d} = \{[i,j,k] \rightarrow [k] \}$ and therefore the reuse space is the 2-dimensional $ij$-plane, $\mathcal{R} = \{[i,j,k] \mid k=0\}$.
This set has four inequality constraints: $c_{1}$ ($i \leq N$), $c_{2}$ ($0 \leq j$), $c_{3}$ ($0 \leq k$), and $c_{4}$ ($k \leq i-j$).
Geometrically, the shape of $\mathcal{D}$ is a tetrahedron with 4 faces (2-faces), 6 edges (1-faces), and 4 vertices (0-faces).
The face lattice of $\mathcal{D}$ is illustrated in Figure~\ref{fig:face-lattice}.
The numbers in each node in the lattice denote which constraints are saturated.
For example, the 2-face labeled ``$\{1\}$'' represents the 2-face obtained by saturating $c_{1}$ (i.e., the triangular face at $i=N$).
Each edge in the lattice can be thought of as saturating one additional constraint.

\subsection{Single-Step Simplification} \label{sec:background-single-step-simplification}
\begin{figure}[tbh]

    \centering
    \begin{subfigure}{0.35\textwidth}
        \centering
        \includegraphics[width=\textwidth]{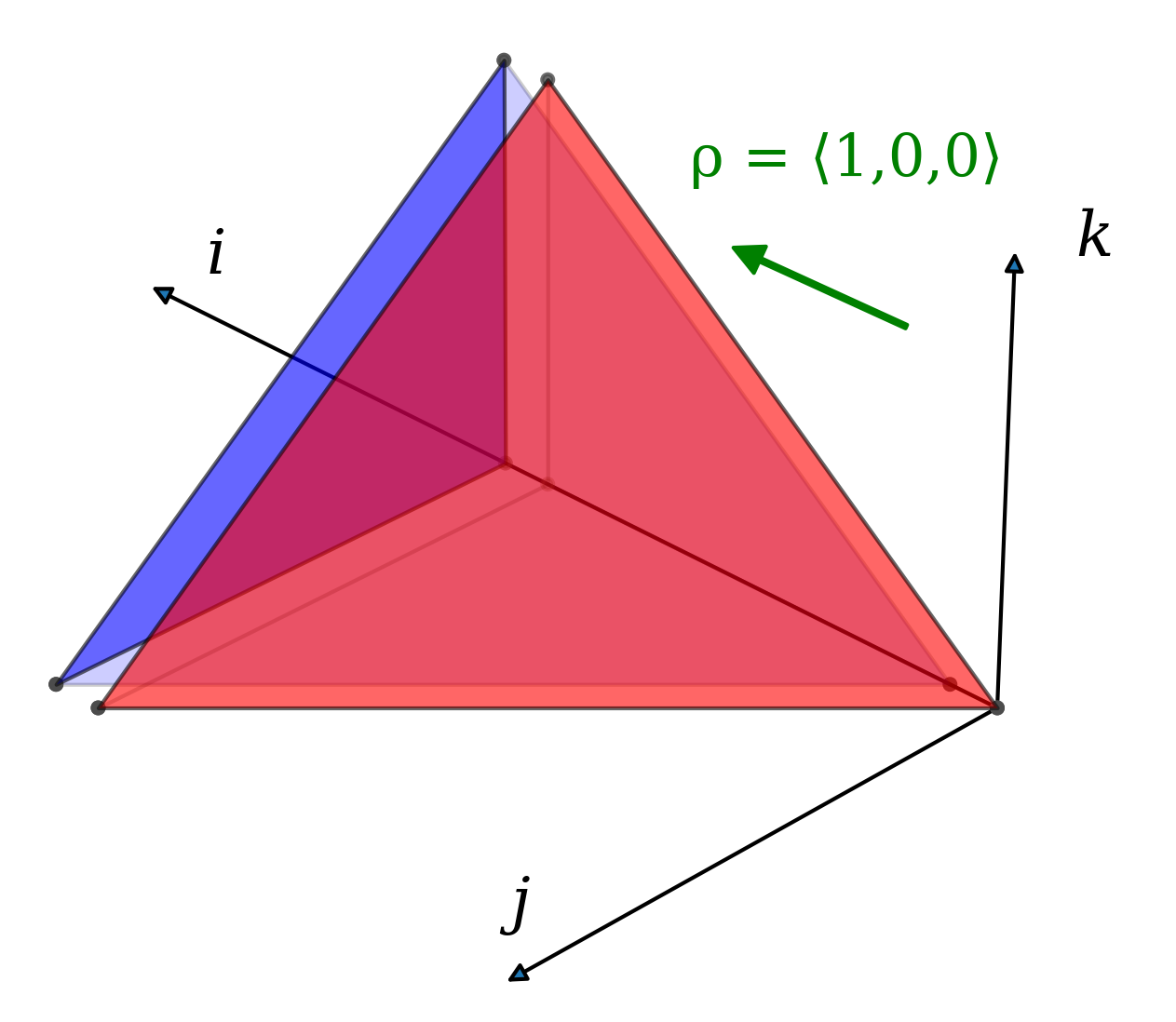}
        \caption{Translate $\mathcal{D}$ (shift red pyramid) by $\halfvec{\rho}$ (green arrow) yielding $\mathcal{D}_s$ (blue pyramid).}
    \end{subfigure}%
    \hspace{5mm}
    \begin{subfigure}{0.4\textwidth}
        \centering
        \includegraphics[width=\textwidth]{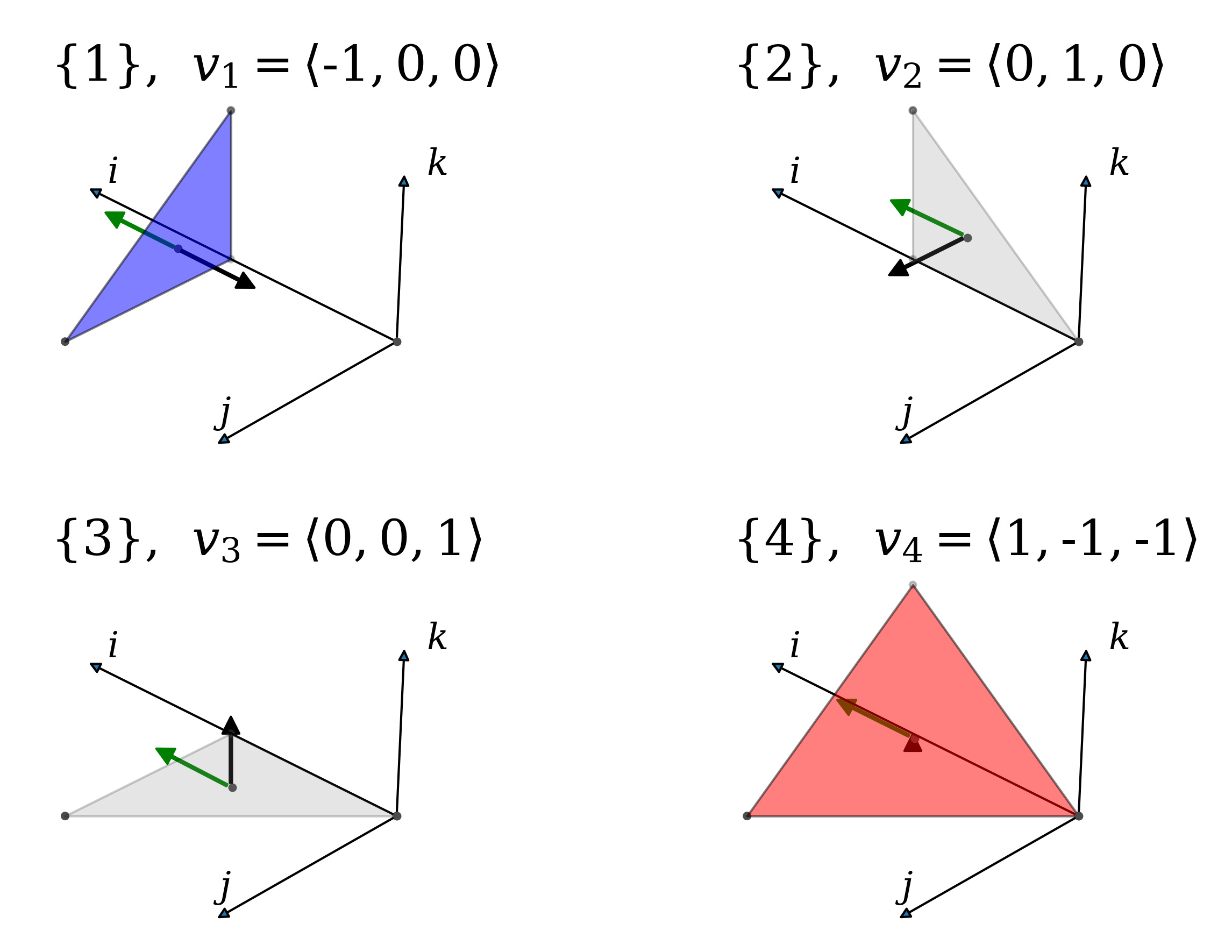}
        \caption{Evaluate the residual computation on only the facets of $\mathcal{D}$ and/or $\mathcal{D}_s$.}
    \end{subfigure}
    \caption{Simplification of a cubic equation (computation defined over a tetrahedron) to a quadratic complexity (the residual computation is defined only over the points in the top red triangular face).}
    
    \Description{Cubic to quadratic complexity.}
    \label{fig:sr-example}
\end{figure}

Our working example involves a reduction with the addition operator.
The $i$'th result, $Y_{i}$, is the accumulation of the values of the reduction body at points in the $i$'th triangular slice of the tetrahedron.
The complexity of the equation is the number of integer points in
$\mathcal{D}$, namely $O(N^3)$.
Simplification is possible because the body has redundancy along any vector $\halfvec{\rho} \in \mathcal{R}$, such as $\halfvec{\rho} = \langle 1, 0, 0 \rangle$ (green arrow) shown in Figure~\ref{fig:sr-example}.
Any two points $[i, j, k], [i', j', k'] \in \mathcal{D}$ separated by a scalar multiple of $\halfvec{\rho}$ read the same value of $X$ because $X_{k} = X_{k'}$.
In other words, the body expression evaluates to the same value at all points along $\langle 1, 0, 0 \rangle$.  
Simplification exploits this reuse to read and compare only the $O(N^2)$ distinct values.
A geometric explanation of simplification is: 
\begin{enumerate}
    \item  Translate $\mathcal{D}$ (shift the red pyramid) by $\halfvec{\rho}$ (green arrow) yielding $\mathcal{D}_s$ (the blue pyramid).
    \item Delete all computations in the intersection of the two.
    \item Evaluate the residual computation on only (a subset of) the facets (here, 2-faces) of $\mathcal{D}$ and/or $\mathcal{D}_s$.
\end{enumerate}

Additionally, some of these facets can be ignored.
This is a \textbf{\textit{critical}} component of how we will construct splits and is discussed further in Sections~\ref{sec:boundary-facets} and~\ref{sec:invariant-facets}.
For now, note that the grey triangular 2-faces ``$\{2\}$'' and ``$\{3\}$'', whose normal vectors are orthogonal to $\halfvec{\rho}$, were already included in the intersection, and the back blue 2-face ``$\{1\}$'' one at $i=N+1$ is \emph{external} since it does not contribute to any answer.
This leaves a residual computation on only the top red oblique 2-face ``$\{4\}$''.
Thus, the $O(N^3)$ computation in Equation~\ref{eq:sr-example} can be rewritten as the following $O(N^2)$ equation:
\begin{equation} \label{eq:sr1}
Y_{i} = \left\{
        \begin{array}{lcl} 
            i=0 &:&\displaystyle \sum_{j=0}^{j \leq i} X_{i-j} \\[1mm]
            i > 0 &:& \displaystyle Y_{i-1} + \sum_{j=0}^{j \leq i} X_{i-j}
        \end{array} \right. 
\end{equation}
obtained from the application of Theorem 5 from GR06, which we do not review in detail here.
The intuition is that simplifying Equation~\ref{eq:sr-example} along $\halfvec{\rho} = \langle 1,0,0 \rangle$ \emph{moves} the computation to the ``\{4\}'' face, which is described by the saturated constraint $k=i-j$.
Thus the subscript, $k$, on $X$ from Equation~\ref{eq:sr-example} is expressed as $X_{i-j}$ in Equation~\ref{eq:sr1}.

\subsection{Equivalence Partitioning of Infinite Choices} \label{sec:background-equivalence-classes}

In this example, the reuse space is multi-dimensional, which means there are infinitely many choices of reuse along which simplification can be performed.
However, not all choices of reuse need to be explored.
The intuition is that any two reuse vectors resulting in the same combination of residual computation can be viewed as the same candidate choice.
In the previous section, the choice of $\halfvec{\rho} = \langle 1,0,0 \rangle$ resulted in a residual computation on only the top oblique 2-face.
GR06 refers to the subset of $\halfvec{\rho}$ that results in the same residual computation as an \textit{equivalance class}.
Then, the set of all equivalence classes can be explored with dynamic programming.
This guarantees the final simplified program's optimality as long as a single reuse vector from each equivalence class is considered.

\begin{figure}[h]
    \centering
    \includegraphics[width=0.8\textwidth]{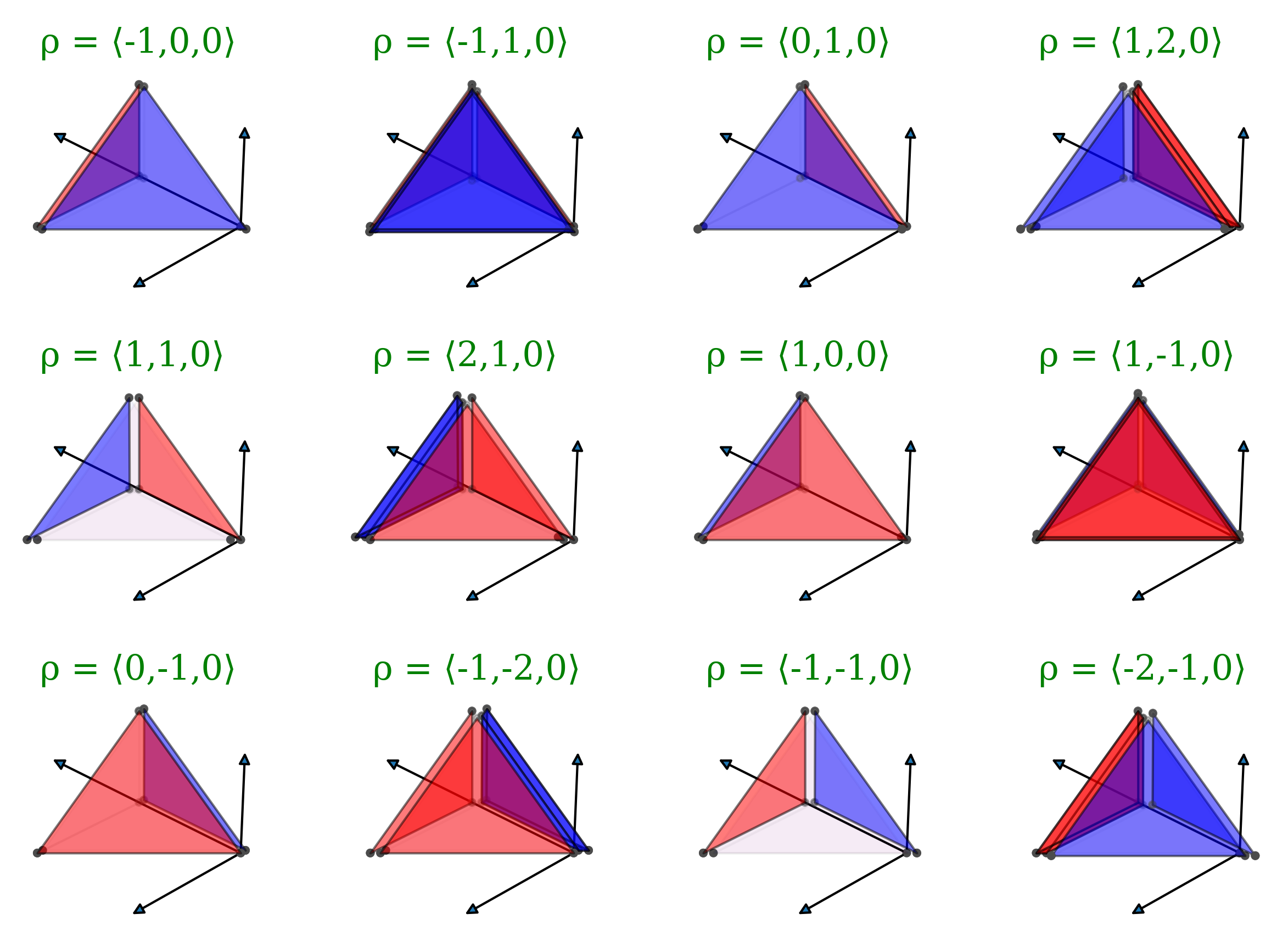}
    \caption{The 12 equivalence classes for the single-step simplification of Equation~\ref{eq:sr-example}.}
    \Description{Equivalence classes}
    \label{fig:equivalence-classes}
\end{figure}

Let $\mathcal{F}$ be an arbitrary facet of the reduction body.
Let $\mathcal{F}_{i}$ denote the $i$'th facet of $\mathcal{F}$ (i.e., its $i$'th child), and let $\halfvec{\nu_{i}}$ be the linear part of the normal vector of $\mathcal{F}_{i}$.
Let the symbol $\oplus$ be the reduction operator, and $\ominus$ be its inverse if $\oplus$ is invertible.
The single-step simplification of $\mathcal{F}$, summarized in Section~\ref{sec:background-single-step-simplification}, with the reuse vector $\halfvec{\rho}$, results in a residual computation on some of its facets.
The orientation of $\halfvec{\rho}$ relative to each facet dictates the type of residual computation that occurs on $\mathcal{F}_{i}$.
There are three possibilities, depending on the sign of the dot product between $\halfvec{\rho}$ and $\halfvec{\nu_{i}}$:
\begin{enumerate}
    \item If $\halfvec{\rho} \cdot \halfvec{\nu_{i}} > 0$ then $\mathcal{F}_{i}$ contributes with the $\oplus$ operator.
    \item If $\halfvec{\rho} \cdot \halfvec{\nu_{i}} < 0$ then $\mathcal{F}_{i}$ contributes with the $\ominus$ operator.
    \item If $\halfvec{\rho} \cdot \halfvec{\nu_{i}} = 0$ then $\mathcal{F}_{i}$ does not contribute at all.
\end{enumerate}
Each facet $\mathcal{F}_{i}$ may be labeled as either an $\oplus$-face, $\ominus$-face, or $\oslash$-face respectively.~\footnote{Labels are pronounced ``positive'', ``negative'', and ``invariant'' faces respectively. The notion of \emph{labelings} are not part of the prior work of GR06, we introduce this language to help explain why simplification can fail in the following sections.}
We say that $\halfvec{\rho}$ \textit{induces a particular labeling}, $\mathcal{L}$, on the facets of $\mathcal{F}$.
Each labeling corresponds to one of the equivalence classes.
For a face with $m$ facets, there are a total of $3^m$ distinct labelings.
However, many of these will not be possible.
There is no way to mark all facets with the same label, for example, since the set of $\halfvec{\rho}$ with a positive dot product to all normal vectors in a convex polytope is empty.
In this example, the 12 possible labelings (equivalence classes) are shown in Figure~\ref{fig:equivalence-classes} with $\oplus$-faces colored in red, $\ominus$-faces in blue, and $\oslash$-faces uncolored.

\subsection{Recursive Simplification} \label{sec:background-recursive-simplification}
In the general case, reductions have a $d$-dimensional reduction body, an $a$-dimensional accumulation, and an $r$-dimensional reuse space.
The process of Section \ref{sec:background-single-step-simplification} is applied recursively on the face lattice, starting with $\mathcal D$.
At each step, we simplify the facets of the current face $\mathcal F$.
The key idea is that exploiting reuse along $\halfvec{\rho}$ avoids evaluating the reduction expression at most points in $\mathcal F$.
Specifically, let $\mathcal{F}'$ be the translation of $\mathcal F$ along $\halfvec{\rho}$.
Then all the computation in $\mathcal{F}\cap \mathcal{F}'$ is avoided, and we only need to consider the two differences $\mathcal{F}'\backslash \mathcal{F}$ and $\mathcal{F}\backslash \mathcal{F}'$, i.e., the union of some of the facets of $\mathcal{F}$.  

At each recursive step down the face lattice, the asymptotic complexity is reduced by exactly one polynomial degree, as facets of $\mathcal F$ are strictly smaller dimensional subspaces.
Furthermore, at each step, the newly chosen $\halfvec{\rho}$ is linearly independent of the previously chosen ones.
Hence, the method is optimal---all available reuse is fully exploited.
This holds regardless of the choice of $\halfvec{\rho}$ at any level of the recursion, even though there may be infinitely many choices.

Bringing everything together, the residual $O(N^2)$ computation in Equation~\ref{eq:sr1} of our working example,
\begin{equation} \label{eq:residual-2d-reduction}
    Y_{i} = \sum_{j=0}^{j \leq i} X_{i-j}
\end{equation}
can be thought of as a completely new 2-dimensional reduction that may be further simplified.
We do not describe this in detail here, as it is very similar to the prefix sum described in Section~\ref{sec:motiv-prefix-sum}.

\subsection{Simplification Enhancing Transformations} \label{sec:reduction-decomposition}

GR06 proposes several simplification-enhancing transformations that can be used to expand the reuse space.
We briefly summarize one of them, which exploits commutativity, called \textit{reduction decomposition} here, as we rely on this heavily in Section~\ref{sec:dD-reductions}.
Given two functions $f_{p}''$ and $f_{p}'$ such that $f_{p} = f_{p}'' \circ f_{p}'$,  a reduction of the form in Equation~\ref{eq:reduction-def} with multi-dimensional accumulation may be rewritten as the following two reductions,
\begin{align} \label{eq:decomp-def}
Y_{f_{p}''(z)} &= \bigoplus_{z \in \mathcal{D}} Z_{f_{p}'(z)} \\
Z_{f_{p}'(z)} &= \bigoplus_{z \in \mathcal{D}} X_{f_{r}(z)}
\end{align}
with the introduction of a new variable $Z$ to hold partial answers.
We can think of this as decomposing a higher dimensional reduction into a lower dimensional \textit{reduction of reductions}, which is legal because the order of accumulation does not matter. 
This transformation is useful because it affects which facets can be ignored.
We precisely characterize this in Section~\ref{sec:boundary-facets}.

\subsubsection*{An Example} Our working example in Equation~\ref{eq:sr-example} has two dimensions of accumulation (i.e., along $j$ and $k$) because $f_{p} = \{[i,j,k] \rightarrow [i]\}$.
Therefore Equation~\ref{eq:sr-example} could be explicitly written as the following double summation,
\begin{equation}
  \label{eq:sr-example-decomp1}
   Y_{i} = \sum_{j=0}^{i} \sum_{k=0}^{i-j} X_{k}
\end{equation}
where the inner reduction accumulates over $k$.
We have not explicitly separated this into two separate equations. Still, the inner reduction should be interpreted as a reduction with an accumulation characterized by $f_{p}' = \{[i,j,k] \rightarrow [i,j] \}$ (i.e., producing a 2-dimensional intermediate answer) and the outer reduction as one characterized by $f_{p}'' = \{[i,j] \rightarrow [i] \}$.
Note that there are many other decompositions available.
For example, the inner reduction could be expressed over $j$ instead or any linear combination of $j$ and $k$, for that matter.
We discuss methods for constructing $f_{p}'$ in Section~\ref{sec:dD-reductions}.

\section{Splitting to Avoid Simplification Failure} \label{sec:extending-simplification}

In this section, we present the intuition of our main result with a simple example.
The subsequent sections provide the precise formulation and proofs for the general cases.

\subsection{How and When Simplification Can Fail: Optimal vs.\ Maximal}  \label{ref:how-can-fail}

As discussed in the previous section, the GR06 algorithm proceeds recursively down the face lattice, simplifying facets one by one while additional reuse exists.
At each step of the recursion, on a particular face $\mathcal{F}$, several candidate choices of reuse are explored, one for each possible labeling of the facets of $\mathcal{F}$.
For a particular labeling, each facet is treated as either an $\oplus$-face, $\ominus$-face, or $\oslash$-face.
Let us assume now that the reduction operator does not admit an inverse, which means that $\ominus$-faces must be avoided.
Recall that some facets may be ignored (e.g., all but the top red oblique 2-face in Figure~\ref{fig:sr-example}), but some can not; let us refer to these facets that can \emph{not} be ignored as \emph{residual facets}.
For a particular labeling, simplification is impossible when two residual facets exist with opposite labels, one $\oplus$-face and one $\ominus$-face.
If all residual facets are $\oplus$-faces, then simplification can obviously proceed.
Similarly, if they are all $\ominus$-faces, we can simply negate $\halfvec{\rho}$ to view them as $\oplus$-faces.
The single-step simplification of Section~\ref{sec:background-single-step-simplification} fails if all candidate choices of reuse induce labelings that involve residual facets with opposite labels.

\subsubsection{An Example}
Look back at the prefix max from Section~\ref{sec:motiv-prefix-max}.
The reduction body is the triangle, $\mathcal{D} = \{[i,j] \mid 0 \leq j \leq i \leq N\}$ and the reuse space is the 1D space along the $i$-axis, $\mathcal{R} = \{[i,j] \mid j=0 \}$.
This means that there are only two labelings, shown on the left of Figure~\ref{fig:prefix-max-equivalence-classes}.
The labeling induced by $\halfvec{\rho} = \langle 1,0 \rangle$ involves a single residual $\oplus$-face.
In contrast, the labeling induced by $\langle -1,0 \rangle$ involves two residual facets with opposite labels (i.e., one red and one blue coloring).
Since at least one labeling exists with no oppositely labeled residual facets, simplification succeeds and enables us to write Equation~\ref{eq:motiv-prefix-max-simplified-v1}.
Note that the vertical patterned blue facet on the left triangle can be ignored because it is a boundary facet, as discussed in Section~\ref{sec:boundary-facets}.
\begin{figure}[tbh]
    \centering
    \includegraphics[width=0.9\textwidth]{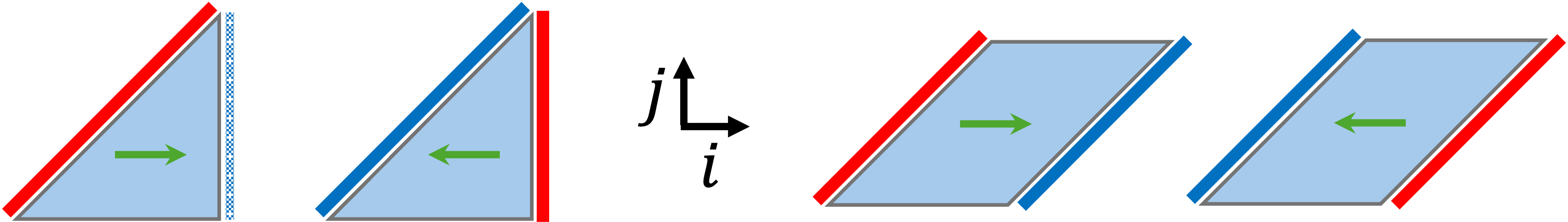}
    \caption{The two labelings for the prefix max example from Section~\ref{sec:motiv-prefix-max} shown by the triangles on the left. The two labelings for the parallelogram in Equation~\ref{eq:parallelogram} are shown on the right. Residual facets are colored in solid red ($\oplus$-faces) and blue ($\ominus$-faces). The vertical patterned blue facet on the left triangle is an outward boundary facet. Therefore, it can be ignored, and thus, the labeling induced by $\langle 1,0 \rangle$ is possible since this does not involve residual facets with opposite labels.}
    \Description{Equivalence classes}
    \label{fig:prefix-max-equivalence-classes}
\end{figure}

Now consider the following example:
\begin{equation} \label{eq:parallelogram}
    Y_{i} = \max_{(0 \leq j) \land (i-N \leq j)}^{(j \leq i) \land (j \leq N)} X_{j}
\end{equation}
This reduction has a body in the shape of a parallelogram with two possible labelings, as shown on the right of Figure~\ref{fig:prefix-max-equivalence-classes}.
Since all (two) labelings involve at least one pair of residual facets with opposite labels, simplification fails.

\subsubsection{Avoid Labelings With Oppositely Labeled Residual Facets}
Given this formulation characterizing when simplification fails, our goal is to split the reduction body into pieces so that we can guarantee the existence of at least one labeling in each piece that involves only residual $\oplus$-faces exclusively or $\ominus$-faces, but not both.
Before we can describe how to do this, we need to augment GR06's original characterization of non-residual facets, i.e., the facets that may be ignored, which we will call \textit{boundary} and \textit{invariant} facets.

\subsection{Boundary Facets} \label{sec:boundary-facets}

Let $\mathcal{H}$ denote the linear space of the facet $\mathcal{F}$, defined by GR06 as the intersection of the effectively saturated constraints characterizing $\mathcal{F}$.
An unsaturated constraint, $c_{i}$ in $\mathcal{F}$ is characterized as a \emph{boundary} constraint if the following condition holds:
\begin{equation}
    \mathcal{H} \cap \mathcal{A} \subseteq \mathcal{H} \cap \mathrm{ker}(c_{i}) 
\end{equation}
where $\mathcal{A}$ is the accumulation space defined previously, and $\mathrm{ker}(c_{i})$ is the null space of the linear part of the constraint.
This says that a facet is a boundary if its linear space contains the entire accumulation space, i.e., multiple points \emph{on that face} contribute to one or more answers. 
Boundary facets are useful because any boundary facet labeled as an $\ominus$-face can be ignored since the ``answer(s)'' are not needed.  

We will make an additional distinction on the degree to which a facet is considered a boundary.
Let us further characterize boundary facets as \textit{strong} or \textit{weak} based on the following definitions.

\begin{definition} \label{def:strong-boundary}
 Let a facet of $\mathcal{F}$, defined by the effectively saturated constraint $c_{i}$, be called a \emph{\textbf{strong}} boundary facet if the following condition holds,
 \begin{equation}
 \mathrm{rank}\big(\mathcal{A} \cap \mathrm{ker}(c_{i})\big) = \mathrm{rank}\big(\mathcal{A}\big)
 \end{equation}
\end{definition} \label{def:weak-boundary}
\begin{definition}
 Let a facet of $\mathcal{F}$, defined by the effectively saturated constraint $c_{i}$, be called a \emph{\textbf{weak}} boundary facet if the following condition holds,
 \begin{equation}
 0 < \mathrm{rank}\big(\mathcal{A} \cap \mathrm{ker}(c_{i})\big) < \mathrm{rank}\big(\mathcal{A}\big)
 \end{equation}
\end{definition}
\noindent In other words, a strong boundary facet is one where no other point contributes to the answer(s) to which the points on the facet contribute.
In contrast, a facet whose linear space contains \textit{part of} (i.e., has a non-trivial and incomplete intersection with) the accumulation space is said to be a weak boundary.
Note that these are mutually exclusive; a facet can not be simultaneously strong and weak.

\subsubsection*{An example}

The distinction between strong and weak boundaries only has meaning in 3D and higher.
Look back at the working example from Section~\ref{sec:background-working-ex}.
The accumulation space in this example is the $jk$-plane.
Of its four 2-faces, shown in Figure~\ref{fig:sr-example}, the ``\{1\}'' face, at $i=N$, is a strong boundary because its linear space, the $jk$-plane contains the entire accumulation space.
The other three 2-faces are all weak boundaries because the intersection of their linear spaces with the accumulation space is a 1D subspace.

\subsection{Invariant Facets}  \label{sec:invariant-facets}

GR06 does not explicitly characterize what we will call invariant facets.
We can think of an invariant facet as the dual of a boundary facet, but from the perspective of the reuse space instead of the accumulation space.
Let an unsaturated constraint $c_{i}$ in $\mathcal{F}$ be characterized as an \textit{invariant} constraint if the following condition holds:
\begin{equation}
    \mathcal{H} \cap \mathcal{R} \subseteq \mathcal{H} \cap \mathrm{ker}(c_{i}) 
\end{equation}
where $\mathcal{R}$ is the reuse space.
Like boundary facets, invariant facets are useful because the recursion never proceeds into an invariant facet (i.e., invariant facets are always labeled as $\oslash$-faces regardless of the choice of reuse $\halfvec{\rho}$).

Similarly, we distinguish the extent to which a facet is invariant based on the following definitions.
\begin{definition} \label{def:strong-invariant}
 Let a facet of $\mathcal{F}$, defined by the effectively saturated constraint $c_{i}$, be called a \emph{\textbf{strong}} invariant facet if the following condition holds,
 \begin{equation}
 \mathrm{rank}\big(\mathcal{R} \cap \mathrm{ker}(c_{i})\big) = \mathrm{rank}\big(\mathcal{R}\big)
 \end{equation}
\end{definition} \label{def:weak-invariant}
\begin{definition}
 Let a facet of $\mathcal{F}$, defined by the effectively saturated constraint $c_{i}$, be called a \emph{\textbf{weak}} invariant facet if the following condition holds,
 \begin{equation}
 0 < \mathrm{rank}\big(\mathcal{R} \cap \mathrm{ker}(c_{i})\big) < \mathrm{rank}\big(\mathcal{R}\big)
 \end{equation}
\end{definition}
Note that a facet may be simultaneously a weak invariant and a weak boundary.
This happens when the intersection of accumulation space, reuse space, and linear space of the facet is non-trivial.

\subsubsection*{An example}

Again, looking back at the working example in Section~\ref{sec:background-working-ex}, the reuse space is the $ij$-plane.
The bottom ``\{2\}'' face is,, therefore,, a strong invariant facet, and the other faces are weak invariant facets.
In other words, any vector $\halfvec{\rho} = \langle \rho_{i}, \rho_{j}, 0 \rangle \in \mathcal{R}$ labels the bottom 2-face an $\oslash$-face because any such $\halfvec{\rho}$ is orthogonal to the normal vector of the bottom 2-face, $\langle 0,0,1 \rangle$.

\subsection{Residual Facets}

As mentioned in Section~\ref{ref:how-can-fail}, we can ignore some of the facets during single-step simplification.
Let us refer to the facets which can \textbf{not} be ignored as residual facets, which are defined as follows:

\begin{definition}
    A facet of $\mathcal{F}$, defined by the effectively saturated constraint $c_{i}$, will be called a \textbf{\textit{residual}} facet if the following condition holds,
    \begin{equation}
        \Big(\mathrm{rank}\big(\mathcal{A} \cap \mathrm{ker}(c_{i})\big) < \mathrm{rank}\big(\mathcal{A}\big)\Big) \land \Big(\mathrm{rank}\big(\mathcal{R} \cap \mathrm{ker}(c_{i})\big) < \mathrm{rank}\big(\mathcal{R}\big)\Big)
    \end{equation}
\end{definition}

\noindent In other words, residual facets are neither strong boundary facets nor strong invariant facets.
Residual facets are those into which the recursion may need to explore that can potentially be labeled as $\ominus$-faces.

\subsection{Split Reduction} \label{sec:split-reduction}

Le Verge~[\citeyear{le_verge_environnement_1992}] showed that a reduction in the form of Equation~\ref{eq:reduction-def}, with the body $\mathcal{D}$ that can be partitioned into  disjoint subsets $\mathcal{D}_{1}$ and $\mathcal{D}_{2}$ (i.e., $\mathcal{D} = \mathcal{D}_{1} \cup \mathcal{D}_{2} $ and $\mathcal{D}_{1} \cap \mathcal{D}_{2} = \phi$) may be rewritten as the following:
\begin{equation} \label{eq:split-reduction}
Y_{f_{p}(z)} = \left\{  \begin{array}{lcl}
      f_p(\mathcal{D}_1) \cap f_p(\mathcal{D}_2) &:&\displaystyle \Big(\bigoplus_{z \in \mathcal{D}_1} X_{f_d(z)}\Big) \oplus \Big(\bigoplus_{z \in \mathcal{D}_2} X_{f_d(z)}\Big) \\[5mm]
      f_p(\mathcal{D}_1) ~\backslash~ f_p(\mathcal{D}_2) &:&\displaystyle \bigoplus_{z \in \mathcal{D}_1} X_{f_d(z)} \\[5mm]
      f_p(\mathcal{D}_2) ~\backslash~ f_p(\mathcal{D}_1) &:&\displaystyle \bigoplus_{z \in \mathcal{D}_2} X_{f_d(z)}
 \end{array}\right.
\end{equation}
Even though the domain is only split into \textit{two} pieces, there may be \textit{three} branches because the pieces may overlap in the answer space.

\subsubsection*{An example}

\begin{figure}[h]
    \centering
    \includegraphics[width=0.6\textwidth]{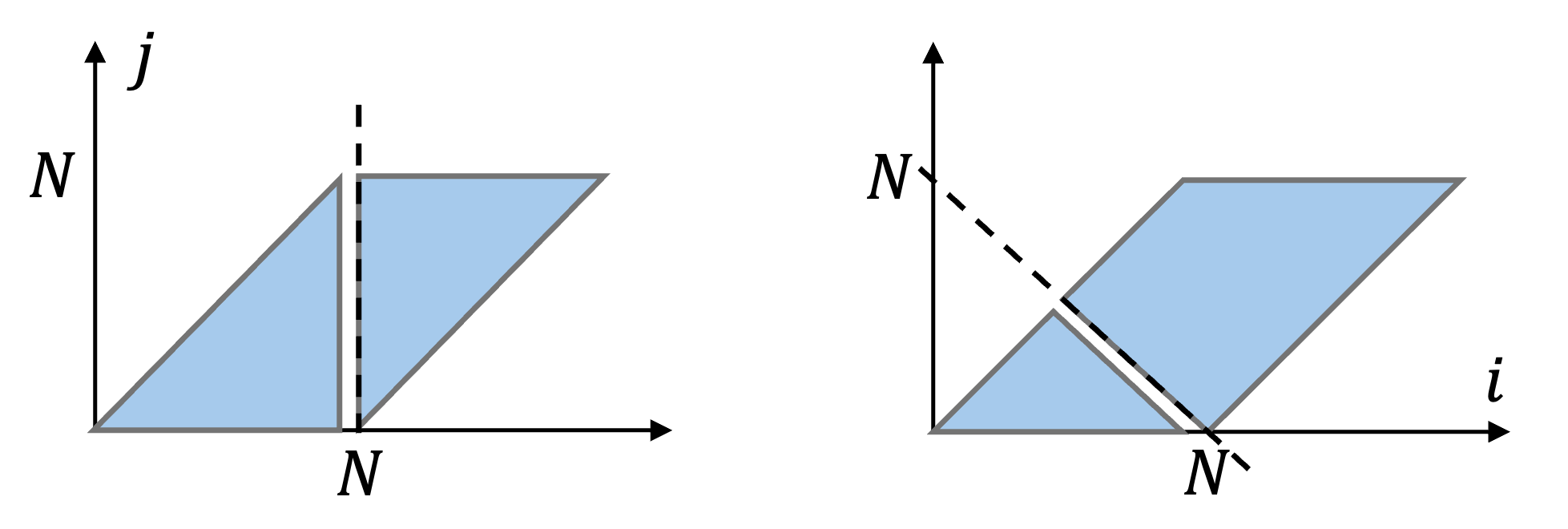}
    \caption{Two example splits of the reduction body of Equation~\ref{eq:parallelogram}, the split along $i=N$ shown on the left and the split along $i+j=N$ on the right. Both cases have two pieces, but the split by $i+j=N$ results in three branches of the transformed equation because the pieces overlap when projected onto the answer space (i.e., the $i$-axis).}
    \Description{split reduction}
    \label{fig:split-reduction}
\end{figure}
The reduction body in the example from Equation~\ref{eq:parallelogram} is the parallelogram, $\mathcal{D} = \{[i,j] \mid (0 \leq j \leq N) \land (i-N \leq j \leq i)\}$.
Consider one split of $\mathcal{D}$ by the hyperplane $i=N$, where $\mathcal{D}_{1} = \mathcal{D} \cap \{[i,j] \mid i<N\}$ and $\mathcal{D}_{2} = \mathcal{D} \cap \{[i,j] \mid i \geq N\}$ shown on the left of Figure~\ref{fig:split-reduction}.
Spliting the reduction in Equation~\ref{eq:parallelogram} along $i=N$ yields:
\begin{equation} \label{eq:split-i-eq-N}
    Y_{i} = \left\{  \begin{array}{lcl}
      0 \leq i<N &:&\displaystyle \max_{j=0}^{j \leq i} X_{j} \\[5mm]
      N \leq i \leq 2N &:&\displaystyle \max_{j=i}^{j \leq N} X_{j}
 \end{array}\right.
\end{equation}
Equation~\ref{eq:split-i-eq-N} only has two branches because the projections of the pieces onto the answer space do not overlap (i.e., project the triangles on the left of Figure~\ref{fig:split-reduction} down onto the $i$-axis).

\newcommand{\lpiece}[0]{\max_{j \in \mathcal{D}_{1}} X_{j}}
\newcommand{\rpiece}[0]{\max_{j \in \mathcal{D}_{2}} X_{j}}

Consider another split $\mathcal{D}$ by the hyperplane $i+j=N$, where $\mathcal{D}_{1} = \mathcal{D} \cap \{[i,j] \mid i+j<N\}$ and $\mathcal{D}_{2} = \mathcal{D} \cap \{[i,j] \mid i+j \geq N\}$ shown on the right of Figure~\ref{fig:split-reduction}, yielding:
\begin{equation} \label{eq:split-i-plus-j-eq-N}
    Y_{i} = \left\{  \begin{array}{lcl}
      (N \leq 2i) \land (i<N) &:&\displaystyle \max\Big(\Big( \lpiece \Big), \Big( \rpiece \Big) \Big) \\[5mm]
      0 \leq 2i <N &:&\displaystyle \lpiece \\[5mm]
      N \leq i \leq 2N &:&\displaystyle \rpiece
 \end{array}\right.
\end{equation}
In this case, there are three branches because the triangle and quadrilateral on the right of Figure~\ref{fig:split-reduction} overlap from $N/2 \leq i < N$ when projected onto the $i$-axis.

\subsection{Strong Boundary and Invariant Splits} \label{sec:boundary-invariant-splits}

In general, there are infinitely many ways to split the reduction body, $\mathcal{D}$, but not all are useful.
Recall that simplification can fail when the reduction operator does not admit an inverse because a pair of oppositely labeled residual facets may be unavoidable.
Splitting isolates some of the facets of $\mathcal{D}$ into one piece, making it possible to separate pairs of conflicting facets.
Each piece retains some facets of $\mathcal{D}$ and obtains a single new facet characterized by the splitting hyperplane.
We want to split $\mathcal{D}$ so that each piece's reuse space is less constrained, thereby enabling simplification.

However, an arbitrary split on $\mathcal{D}$ may undesirably further constrain the available reuse space of each piece because the newly introduced facet changes the set of new labelings permissible in each piece.
However, we can avoid this issue by only making splits that introduce a new strong boundary or strong invariant facet.
This is useful because strong boundary $\ominus$-faces and invariant facets are never residual and can, therefore, be ignored.

If $\mathcal{D}$ is $d$-dimensional, then let $h$ be a ($d-1$)-dimensional hyperplane that separates $\mathcal{D}$ into two non-empty pieces.
The hyperplane $h$ is characterized by a single equality constraint, $c_{h}$.
Let $\mathrm{ker}(c_{h})$ denote the null space of the linear part of $c_{h}$.
\begin{definition} \label{def:strong-boundary-split}
Let $h$ be called a strong boundary split if $\mathcal{A} \subseteq \mathrm{ker}(c_{h})$.
\end{definition}
\begin{definition} \label{def:strong-invariant-split}
Let $h$ be called a strong invariant split if $\mathcal{R} \subseteq \mathrm{ker}(c_{h})$.
\end{definition}

For example, the vertical split at $i=N$ shown previously in Figure~\ref{fig:split-reduction} illustrates a strong boundary split because the accumulation space of the reduction in Equation~\ref{eq:parallelogram} is $\mathcal{A} = \{[i,j] \mid i=0\}$, which is indeed a subset of the null space of the linear part of $i=N$ (they happen to be the same space in this example).
Consequently, the resultant reductions in each branch of Equation~\ref{eq:split-i-eq-N} can be simplified independently because there exists a labeling involving only residual $\oplus$-faces for each.
These are indeed both instances of the prefix max discussed in Section~\ref{sec:motiv-prefix-max}.

\section{Problem Formulation and Hypotheses} \label{sec:problem-formulation}

As we have discussed, simplification is a powerful transformation that lowers the asymptotic complexity of the underlying computation.  
However, it is not always possible as motivated by the reduction examples in Equations~\ref{eq:motiv-i-to-2i} and~\ref{eq:parallelogram}.
In this section, we make our primary claim in the form of Theorem~\ref{th:main-claim}.
Like GR06, we assume that the input program only involves a single size parameter and any reductions present are \textit{independent} (i.e., there is no cycle in the dependence graph among the variables appearing inside the reduction body and the answer variable on the left-hand side of the container equation).

\begin{theorem}
\label{th:main-claim}
Given an independent reduction with a $d$-dimensional body, an $a$-dimensional accumulation space, and an $r$-dimensional reuse space, it may always be transformed into an equivalent reduction with an asymptotic complexity that has been decreased by $l = \mathrm{min}(a, r)$ polynomial degrees.
\end{theorem}

The proof of Theorem~\ref{th:main-claim} will follow from Sections~\ref{sec:dD-reductions} and~\ref{sec:2D-reductions (triangles)} where we show how to handle all possible combinations of the dimensionalities of the accumulation space, reuse space, and reduction body that can occur.
For each case, we will show that the reduction can be split into pieces such that each piece has at least one possible labeling without any residual $\ominus$-faces.

\subsection{Assumptions}

We make several assumptions about the input reductions to justify Theorem~\ref{th:main-claim}.
We emphasize that this does not introduce any loss of generality and the following assumptions are made solely to facilitate the proofs in the following sections.

\subsubsection{Separate Accumulation and Reuse Dimensions Only}

We must only consider reductions involving separate accumulation and reuse dimensions, where $a + r = d$.
Any reduction where $a+r \neq d$ can be systematically transformed into one or more instances of reductions where $a+r = d$.
Therefore, it is sufficient to only consider reductions with accumulation and reuse dimensions.

First, consider the case where $a+r>d$.
In such cases, the accumulation and reuse spaces have a non-trivial intersection, which means that the reduction accumulates the \textit{same value} at many points in the body (i.e., along this intersection).
GR06 describes special simplification cases that may be applied to remove this intersection, which they call \textit{higher order operator} and \textit{idempotence} simplifications.
The working example from Section~\ref{sec:background-working-ex} is in an instance of this case; the reduction body is 3-dimensional while the accumulation and reuse spaces are 2-dimensional, though we did not employ these special simplifications.

Second, consider the case where $a+r < d$.
Such cases can be viewed as families of reductions, with  $d-(a+r)$ \emph{independent parameters}, reading independent slices of the inputs and producing independent slices of the outputs.
For example, the reduction:
\begin{equation*}
    Y[i,j] = \max_{k=i}^{2i} X[j,k]
\end{equation*}
has a 3D domain, 1D accumulation, and 1D reuse space.
However, the index $j$ should be viewed as an independent parameter.
Thus, the overall computation can be viewed as $O(N)$ instances of independent 2D reductions, each with a 1D accumulation and 1D reuse space, one for each value of $j$, embedded in a 3D space.
No simplification is possible among the different instances (each of the  $O(N)$ reduction instances along $j$ must be computed), but further simplification may be possible within the $ik$ dimensions).

\subsubsection{Orthogonal Accumulation and Reuse Along Canonical Axes} \label{sec:assumption-canonic-axes}

Let us assume that the accumulation and reuse spaces are \textit{orthogonal} and are oriented along the canonical axes. 
Let the reuse space be along the first $r$ canonical axes and the accumulation space be along the last $a$ axes.
The program variable domains and the indexing expressions can always be reindexed to put the accumulation and reuse along the canonic axes as described by Le Verge~[\citeyear{le_verge_environnement_1992}].

\subsubsection{Domain of the Reduction Body is a Simplex}

We restrict our analysis to domains that are simplexes (i.e., hyper-triangular) based on the following definition, adapted from Gruber~[\citeyear{gruber_convex_2007}].
\begin{definition} \label{def:simplex}
A ($d$)-simplex is the ($d$)-dimensional polytope defined as the convex combination of $d+1$ affinely independent vertices.
\end{definition}
In practice, decomposing the domain into simplices may not always be necessary.
However, we use properties of simplices to prove that simplification is always possible.
There may be heuristic solutions to decide how to split non-simplices, but we do not consider any such approaches here.
Therefore, in the remaining discussion, we assume that any reductions appearing in the input program have been preprocessed and initially decomposed into simplices.

Our maximal simplification result directly carries over to general polyhedral sets because any ($d$)-dimensional parametric polytope can be decomposed into the union of ($d$)-dimensional simplices.
Triangulating multi-dimensional polytopes is common in computer graphics~\cite{ganapathy_new_1982, narkhede_vii5_1995}, for example.
Simplices have useful properties used in our proofs:
\begin{itemize}
    \item Any ($k$)-face of an ($d$)-simplex is itself a ($k$)-simplex.
The number of ($k$)-faces of a ($d$)-simplex are given by the  binomial coefficients, $d+1 \choose k+1$
unique combinations of its $d+1$ vertices. 
    \item A ($d$)-simplex can be split into two ($d$)-simplices by adding one new affinely independent vertex. See Lemma~\ref{lemma:simplicial-cut}.
\end{itemize}

\subsection{Simplex-Preserving Strong Boundary or Invariant Splits} \label{sec:simplex-preserving-splits}

Repeatedly trying to make arbitrary splits runs the risk of falling into an endless loop.
We need a way to guarantee that the process of splitting will terminate.
Recall that simplification of a reduction with the body $\mathcal{D}$ fails when every possible labeling involves one or more pairs of oppositely labeled \emph{residual} facets.
Therefore, if we can repeatedly split $\mathcal{D}$ in such a way that \emph{both} preserves the total number of facets in each piece \emph{and} strictly decreases the number of residual facets, then it will be possible to guarantee that each piece has a labeling with no conflicting residual facets.

We combine the following two ideas.
First, we will use the fact that splitting a $d$-simplex through any of its ($d-2$)-faces produces two $d$-simplices per Lemma~\ref{lemma:simplicial-cut} and therefore preserves the total number of facets in each piece.
Second, we will use the notion of a strong boundary or invariant split, which introduces a single new \emph{non-residual} facet, as described in Section~\ref{sec:boundary-invariant-splits}.
Combining these two ideas, by making strong boundary or invariant splits through ($d-2$)-faces of simplices, guarantees that the process of splitting will produce pieces with strictly fewer residual facets because the number of total facets in each piece remains the same.
Such splits will be called Simplex-Preserving Strong Boundary (SPB) or Invariant (SPI) splits.

When the accumulation space is one-dimensional, an SPB split can be constructed by infinitely extending one of the ($d-2$) faces along it.
Similarly, when the reuse space is 1-dimensional, extending a ($d-2$)-face along it yields an SPI split.
Since there are finitely many ($d-2$)-faces, there are finitely many candidate SPB and SPI splits to process.

\begin{lemma} \label{lemma:simplicial-cut}
Let a splitting hyperplane of a ($d$)-dimensional polytope be any ($d-1)$-dimensional hyperplane that has points on both sides of the hyperplane.
Any splitting hyperplane that saturates a (d-2)-face of a (d)-simplex produces two (d)-simplices.
\end{lemma}
\begin{proof}
By definition, an ($d$)-simplex is the convex combination of ($d+1$) vertices.
Additionally, every ($k$)-face is itself a ($k$)-simplex, which is just the simplex formed by the ($k+1$) vertices of the ($k$)-face.
Any splitting ($d-1$)-dimensional hyperplane that saturates a ($d-2$)-face contains its ($d-1$) vertices.
There are two remaining vertices, and these, by definition, can not be part of the splitting hyperplane.
This is because the hyperplane is itself a ($d-1$)-simplex and if it contained one of these additional two vertices, then it would saturate an entire ($d-1$)-face of the domain.
Therefore, we can consider these other two vertices as separate from the hyperplane.
We will refer to these as vertex A and vertex B.
The convex combination of vertices A and B forms a (1)-simplex (i.e., a (1)-face or 1-dimensional linear subspace that connects the vertices).
Let point C be any point in this (1)-simplex.
We can construct the following two new sets.
\begin{enumerate}
    \item The convex combination of the ($d-1$) vertices on the ($d-2$)-face, vertex A, and point C.
    \item The convex combination of the ($d-1$) vertices on the ($d-2$)-face, vertex B, and point C. Then take the set difference of this set with the previous set.
\end{enumerate}
Since both are convex combinations of ($d+1$) vertices, they are both, by definition, (d)-simplices.
\end{proof}

\subsubsection{A 2D Example}

The vertical split at $i=N$ illustrated previously in Figure~\ref{fig:split-reduction} is analogous to an SPB split in the sense that it is a strong boundary split through a vertex.
However, a parallelogram is not a simplex, of course.
Regardless, we can compute the constraint $i=N$ as follows.
First, compute the linear space, $H$, of the bottom right vertex of the parallelogram in Figure~\ref{fig:split-reduction}, as the intersection of the saturated constraints describing this vertex:
\begin{equation*}
H = \{ [i,j] \mid (j=0) \land (j=i-N) \}
\end{equation*}
Next, construct the relation characterizing the 1-dimensional basis, $b$, of the accumulation space from the kernel of the write function, $f_{p} = \{[i,j] \rightarrow [i] \}$ in this example, to obtain the relation:
\begin{equation*}
b = \{ [i,j] \rightarrow [i,j+1] \}
\end{equation*}
Then compute its transitive closure, $b^{*}$, to obtain the relation:
\begin{equation*}
b^{*} = \{ [i,j] \rightarrow [i,j'] \}
\end{equation*}
Then apply $b^{*}$ to the linear space of the vertex, $H$, to obtain the set:
\begin{equation*}
\{ [i,j] \mid i=N \}
\end{equation*}
Finally, the single equality constraint is used to characterize the splitting hyperplane $i=N$.
Extending $H$ by $b^{*}$ effectively removes one of the equality constraints because extending it increases its dimensionality by one.
Since we started from a set with two equality constraints (i.e., because $H$ came from a ($d-2$)-face), its extension is guaranteed to have a single equality constraint.
These operations are available in the integer set library, \texttt{isl}~\cite{verdoolaege_isl_2010}.

\subsubsection{A 3D Example}

In three dimensions, a simplex-preserving split is any plane passing through an edge of a tetrahedron.
Any such plane splits it into two tetrahedra (one of which may be empty) as illustrated in Figure~\ref{fig:tetrahedra-splits}.

\begin{figure}[h]
    \centering
    \begin{subfigure}{0.56\textwidth}
        \centering
        \includegraphics[width=0.8\textwidth]{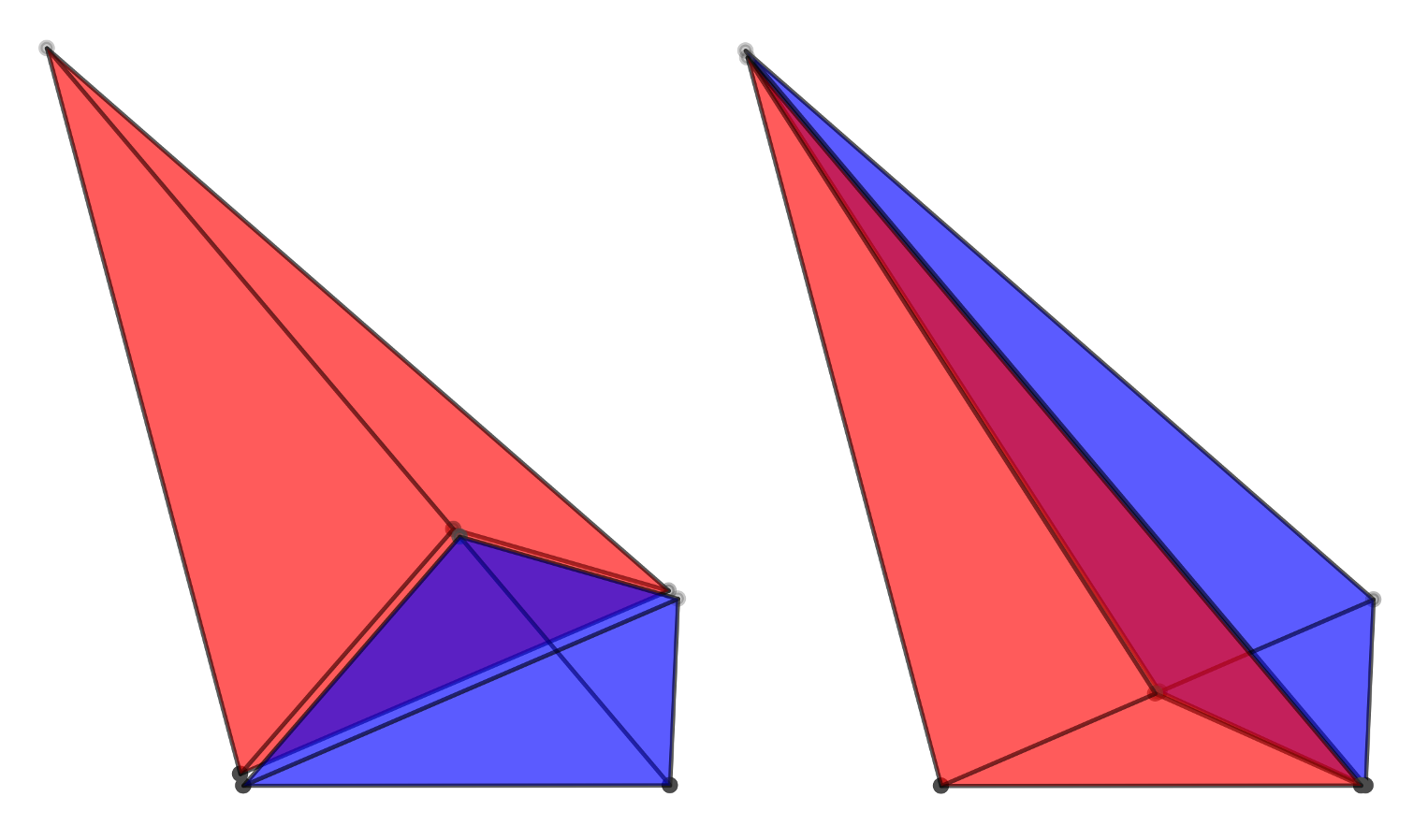}
        \caption{Two boundary splits, each through an edge.}
        \label{fig:tetrahedra-splits-a}
    \end{subfigure}%
    \hspace{5mm}
    \begin{subfigure}{0.2833\textwidth}
        \centering
        \includegraphics[width=0.8\textwidth]{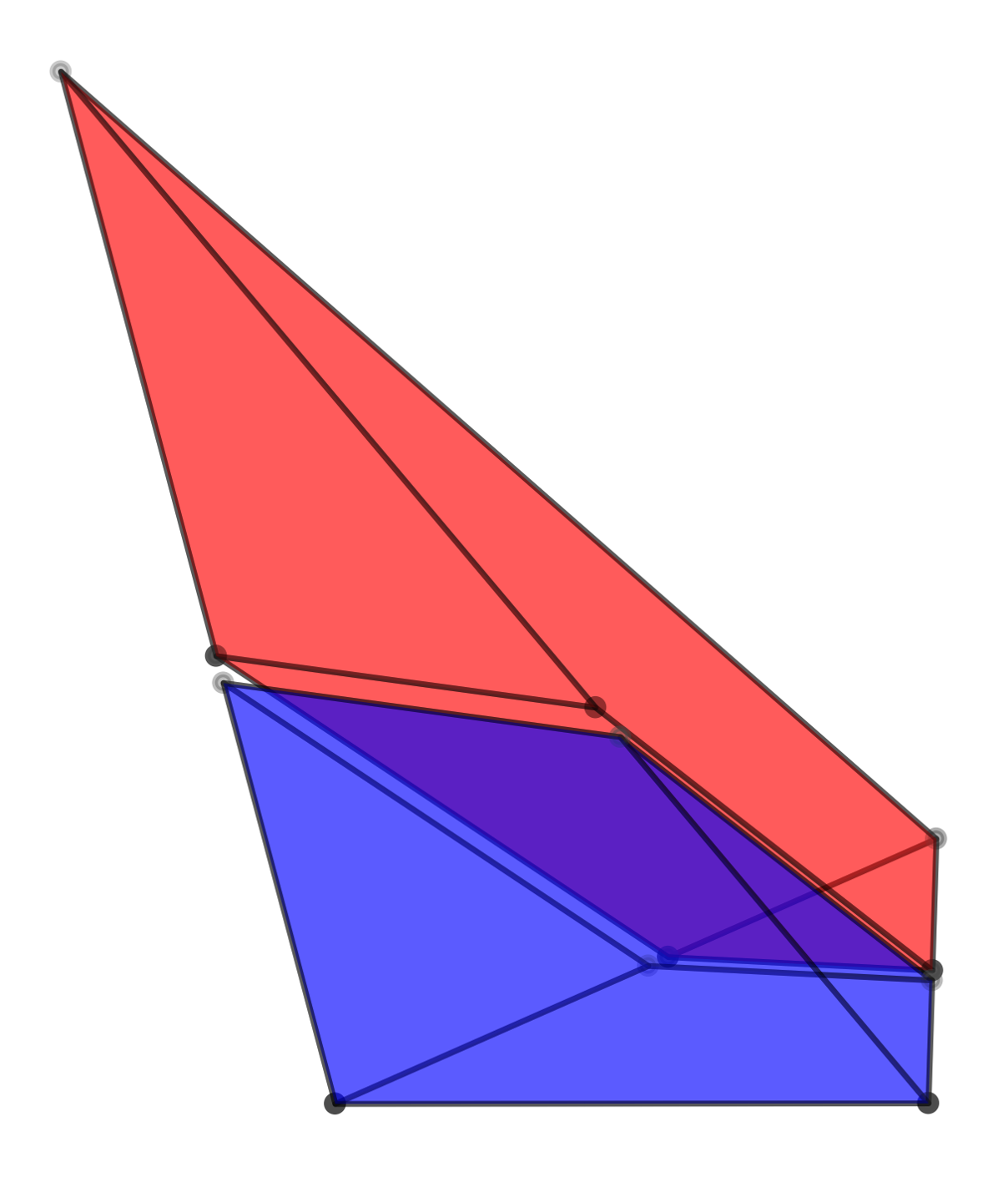}
        \caption{Arbitrary boundary split.}
        \label{fig:tetrahedra-splits-b}
    \end{subfigure}
    \caption{\textit{Tetrahedron}-preserving strong boundary splits (shown on left). Axes are not shown, but imagine that the accumulation space is 1D along the $k$-dimension pointing upwards. This means each split here is a strong boundary split per Section~\ref{sec:boundary-invariant-splits}, but if the split does not pass through an edge, then the resulting pieces are not guaranteed to be tetrahedra, as shown on the right. Splits passing through edges, however, result in two tetrahedra, as shown on the left. }
    
    \Description{Tetrahedra-preserving splits}
    \label{fig:tetrahedra-splits}
\end{figure}

\section{3-Dimensional and Higher Reductions} \label{sec:dD-reductions}

In this section, we consider the single-step simplification of Section~\ref{sec:background-single-step-simplification} for 3-dimensional and higher reductions where all possible labelings involve at least one pair of oppositely labeled facets.
Therefore, splitting the reduction body may be required.
Our goal is to show that pairs of oppositely labeled residual facets can always be avoided.
We will show that all but three residual facets can \textit{either} be transformed into strong boundary facets by reduction decomposition \textit{or} ignored by choosing a reuse vector that labels them as $\oslash$-faces.
Two of the remaining three residual facets may involve opposite labels.
Whenever this happens, we will split the domain so that all $\oplus$-faces are on one side of the split and all $\ominus$-faces are on the other. 
Then, in the piece containing the $\ominus$-face(s), we will negate $\halfvec{\rho}$ to view them as $\oplus$-faces.

We rely on the fact that any ($d-1$)-face has a non-trivial intersection with subsets of the accumulation and reuse space oriented along the canonical axes, per Lemma~\ref{lemma:nti}.

\begin{lemma} \label{lemma:nti}
The linear space of any ($d-1$)-face has a non-trivial intersection with the linear subspace spanned by any two or more canonic axes.
\end{lemma}
\begin{proof}
The ($d-1$)-dimensional linear space describing any ($d-1$)-face involves $d$ indices and one equality constraint.
The linear subspace spanned by $q$ canonic axes involves $d$ indices and $(d-q)$ equality constraints and is $q$-dimensional.
Their intersection involves $(1+d-q)$ equality constraints.
Therefore, it is ($q-1$)-dimensional.
When $q>1$, this intersection is non-trivial.
\end{proof}

\subsection{Avoid Some Residual Facets With Reduction Decomposition}

The reduction decomposition transformation described in Section~\ref{sec:reduction-decomposition} is useful because it has the effect of transforming a weak boundary into a strong boundary.
This happens because the accumulation space of the inner reduction (i.e., the null space of $f_{p}'$ in Equation~\ref{eq:decomp-def}) is strictly smaller after decomposition.
Recall that a weak boundary, per Definition~\ref{def:weak-boundary}, is a facet that partially intersects the accumulation space.
By constructing the inner accumulation space to only involve the dimensions contained by a weak facet, the facet becomes a strong boundary in the inner reduction.
Multiple weak boundary facets can simultaneously be transformed into strong boundary facets by constructing the inner accumulation as the subset of the accumulation space that they share.

\begin{lemma} \label{lemma:projdec}
 Given $a$ dimensions of accumulation, any set of $a$ residual facets can be transformed into $a-1$ strong boundary facets via reduction decomposition.
\end{lemma}
\begin{proof}
Let $[\mathcal{F}_{1}, \mathcal{F}_{2}, ..., \mathcal{F}_{k}]$ be a list of $k$ residual facets.
Take the first residual facet, $\mathcal{F}_{1}$, and let $\mathcal{A}_{1}$ be the intersection of its linear subspace $\mathcal{H}_{1}$ and the $a$-dimensional accumulation space $\mathcal{A}$,
\begin{equation}
    \mathcal{A}_{1} = \mathcal{H}_{1} \cap \mathcal{A}
\end{equation}
$\mathcal{A}_{1}$ is a ($a-1$)-dimensional subspace per Lemma ~\ref{lemma:nti}.
Now let $\mathcal{A}_{p-1}$ denote the intersection of the first ($p-1$) residual faces and the accumulation space,
\begin{equation}
    \mathcal{A}_{p-1} = \mathcal{H}_{1} \cap \mathcal{H}_{2} \cap ... \cap\mathcal{H}_{p-1} \cap \mathcal{A}
\end{equation}
$\mathcal{A}_{p-1}$ is a ($a-p+1$)-dimensional subspace.
Now we decompose the reduction, per Section~\ref{sec:reduction-decomposition}, where the inner reduction's accumulation space is precisely along $\mathcal{A}_{p-1}$.
This is done for up to $a$ residual facets since $\mathcal{A}_{p-1}$ is at least 1-dimensional when $p=a$.
The inner reduction now has a 1D accumulation space, and $p-1$ of the residual facets are subsequently strong boundaries (i.e., non-residual).
\end{proof}

\subsection{Avoid Some Residual Facets With Appropriate Reuse Selection}

Any facets labeled as an $\oslash$-face can be ignored.
Multiple facets can be labeled as an $\oslash$-face by selecting a reuse vector orthogonal to all of the facet normal vectors.
The following proof is analogous to that of Theorem~\ref{lemma:projdec}.

\begin{lemma} \label{lemma:depdec}
 Given $r$ dimensions of reuse, any set of $r-1$ residual facets can be labeled as $\oslash$-faces by choosing a reuse vector in their combined intersection with the reuse space.
\end{lemma}
\begin{proof}
Let $[\mathcal{F}_{1}, \mathcal{F}_{2}, ..., \mathcal{F}_{k}]$ be a list of $k$ residual facets.
Take the first facet $\mathcal{F}_{1}$ and let $\mathcal{R}_{1}$ be the intersection of its linear space $\mathcal{H}_{1}$ with the $r$-dimensional reuse space $\mathcal{R}$,
\begin{equation}
    \mathcal{R}_{1} = \mathcal{H}_{1} \cap \mathcal{R}
\end{equation}
$\mathcal{F}_{1}$ is an ($r-1$)-dimension subspace.
Let $\mathcal{R}_{p-1}$ denote the intersection of the first ($p-1$) residual facets and the reuse space,
\begin{equation}
    \mathcal{R}_{p-1} = \mathcal{H}_{1} \cap \mathcal{H}_{2} \cap ... \cap \mathcal{H}_{p-1} \cap \mathcal{R}
\end{equation}
$\mathcal{R}_{p-1}$ is a ($r-p+1$)-dimensional subspace.
Any reuse vector $\halfvec{\rho} \in \mathcal{R}_{p-1}$ is orthogonal to the normal vectors of all $p-1$ facets and therefore labels all as $\oslash$-faces.

\end{proof}

\subsection{All Possible Scenarios}

Now consider an arbitrary reduction over a $d$-dimensional simplicial domain with $a$ dimensions of accumulation and $r$ dimensions of reuse.
There are $d+1$ facets on the reduction body because it is a simplex.
We just showed, in the previous subsections, how to avoid up to $(a-1)+(r-1) = d-2$ of them using reduction decomposition and an appropriate choice of reuse.
Consequently, there are only two possible scenarios that must be considered.

\subsubsection{Base Case: $d-1$ or Fewer Residual Facets}

In this case, $a-1$ residual facets can be transformed into strong boundaries per Lemma ~\ref{lemma:projdec} and $r-1$ facets can be labeled as $\oslash$-faces per Lemma ~\ref{lemma:depdec}.
This leaves at most one remaining residual facet because $(d-1)-(a-1)-(r-1) \leq 1$.
At least two residual facets must be present for simplification to fail; therefore, simplification succeeds in this case.

\subsubsection{General Case: $d$ or More Residual Facets} \label{sec:general-dD-case}
Like the base case, $a-1$ residual facets can be transformed into strong boundaries, and $r-1$ facets can be labeled $\oslash$-faces.
In this case, however, there can be up to three remaining residual facets since $(d+1)-(a-1)-(r-1) \leq 3$.
Among the remaining residual facets, two of them may involve opposite labels.
In other words, there may be one $\oplus$-face and two $\ominus$-faces, or two $\oplus$-faces and one $\ominus$-face.
In this case, a single SPB or SPI split can be made through one of their ($d-2$)-faces to separate the conflicting facets.

\section{2-Dimensional Reductions: Fractal Simplification of Triangles} \label{sec:2D-reductions (triangles)}

Only one type of reduction can occur in two dimensions: a reduction with a 1D accumulation and 1D reuse space.
Per Section~\ref{sec:assumption-canonic-axes}, let the reuse space be oriented horizontally along the $i$-axis and accumulation vertically along the $j$-axis.
This means that vertical and horizontal edges are boundary and invariant edges.

There are only three types of triangles that can occur:
\begin{enumerate}
    \item 1 residual edge (i.e., a right triangle, which is an instance of a standard scan, e.g., Section~\ref{sec:motiv-prefix-max})
    \item 2 residual edges (some of which require fractal simplification)
    \item 3 residual edges, (can be split into two disjoint instances, each with 2 residual edges)
\end{enumerate}

\noindent The following sections describe how to make 2-dimensional versions of SBP or SBI splits, previously described in Section~\ref{sec:simplex-preserving-splits}, to obtain instances of the first case.

\subsection{Base Case: Right Triangles}

Any triangle with only one residual edge can always be simplified.  
If the residual edge monotonically increases, we exploit reuse along $\halfvec{\rho} = \langle 1,0\rangle$.  
This yields a standard scan (e.g., prefix-max, prefix-min, etc.).  
Otherwise, we exploit reuse along $\langle -1,0\rangle$, producing a backward scan (suffix-max, suffix-min, etc.)

\subsection{Two Residual Edges}

In this case, there is either one boundary (vertical) or one invariant (horizontal)  edge.  Let the three vertices of the triangle be,  $\{\mathcal{V}_{0}, \mathcal{V}_{1}, \mathcal{V}_{2}\}$, of which, $\mathcal{V}_{0}$, which we call the corner vertex, is the intersection of the two residual edges.  There are two sub-cases depending on the relative orientation of the vertices.

\subsubsection{Covered Corner} \label{sec:covered}

Let the corner vertex (the vertex between the two residual edges) be called \textit{covered} if it lands between the other two vertices when all vertices are projected onto either of the axes.
As illustrated by the red point in Figure~\ref{fig:fig:third-vertex-covered}, and by the definition of $\mathcal{V}_0$ being covered, its projection on the other edge is a point on that edge.  Hence, an SPB or SPI split through  $\mathcal{V}_{0}$ yields two right triangles.  
One can be simplified as a forward-scan and the other as a backward-scan.

\subsubsection{Corner Vertex Is Not Covered} \label{sec:uncovered}
There is no loss of generality in assuming that the entire triangle is in the positive quadrant and the corner vertex is at the origin (this can be accomplished by a simple reindexing of one or both of the input/output arrays).  Let us first consider that the non-residual edge is vertical, and let $\mathcal{V}_1$ be \emph{below} $\mathcal{V}_2$.  If $n$ is sufficiently large, we make one horizontal cut through the lower vertex, $\mathcal{V}_1$, and let $\mathcal{V}_2'$ be its intersection with the upper residual edge.  Next, we make a vertical cut through $\mathcal{V}_2'$, intersecting the lower residual edge at, say, $\mathcal{V}_1'$.  We now have three triangles,
$\Delta_1= [\mathcal{V}_0, \mathcal{V}_1',\mathcal{V}_2']$,
$\Delta_2= [\mathcal{V}_1',\mathcal{V}_2', \mathcal{V}_1]$
and $\Delta_3= [\mathcal{V}_2', \mathcal{V}_1, \mathcal{V}_2]$.  Of these, the latter two are right triangles, and $\Delta_1$ is geometrically similar to the original triangle as illustrated on the left side of Figure~\ref{fig:fig:third-vertex-uncovered}.

\begin{figure*}[h]
  \centering
  \includegraphics[width=0.8\linewidth]{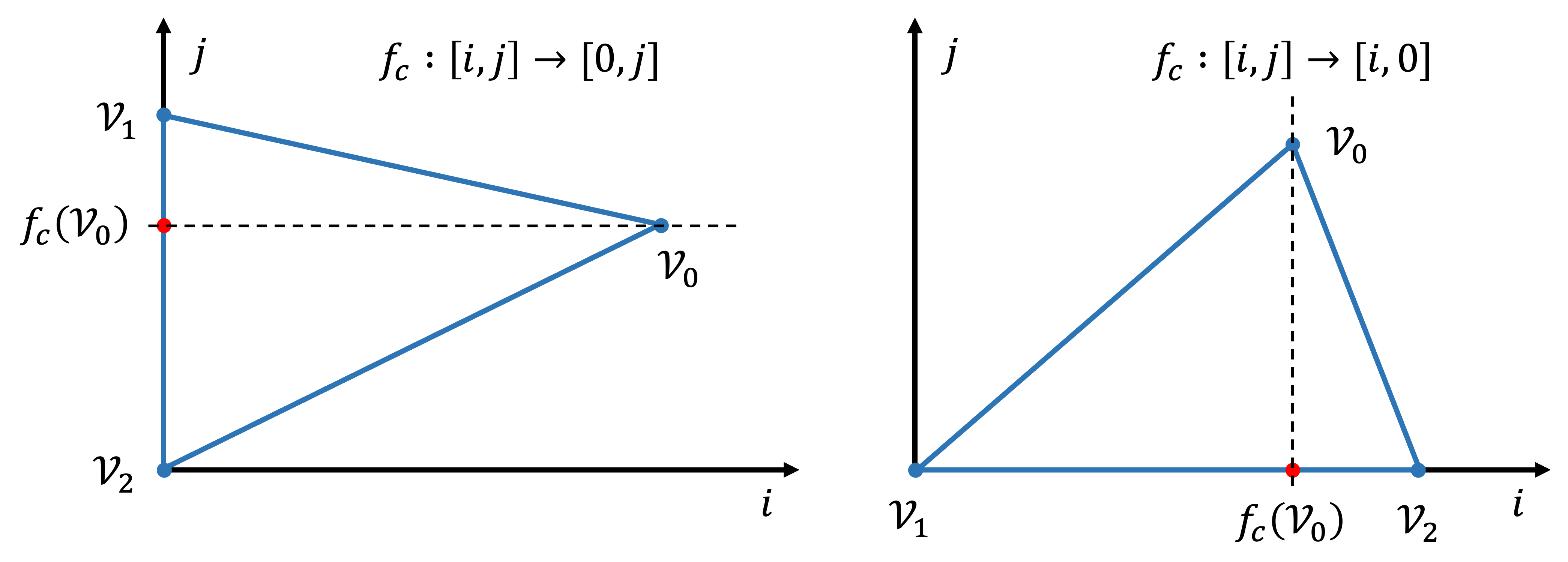}
\caption{Covered corner vertex:  The projection of  $\mathcal{V}_0$ (shown by the red point) is between $\mathcal{V}_1$ and$\mathcal{V}_2$.  A single horizontal or vertical cut through the corner produces two right triangles.}
\label{fig:fig:third-vertex-covered}
\end{figure*}

\begin{figure*}[h]
  \centering
\includegraphics[width=0.8\linewidth]{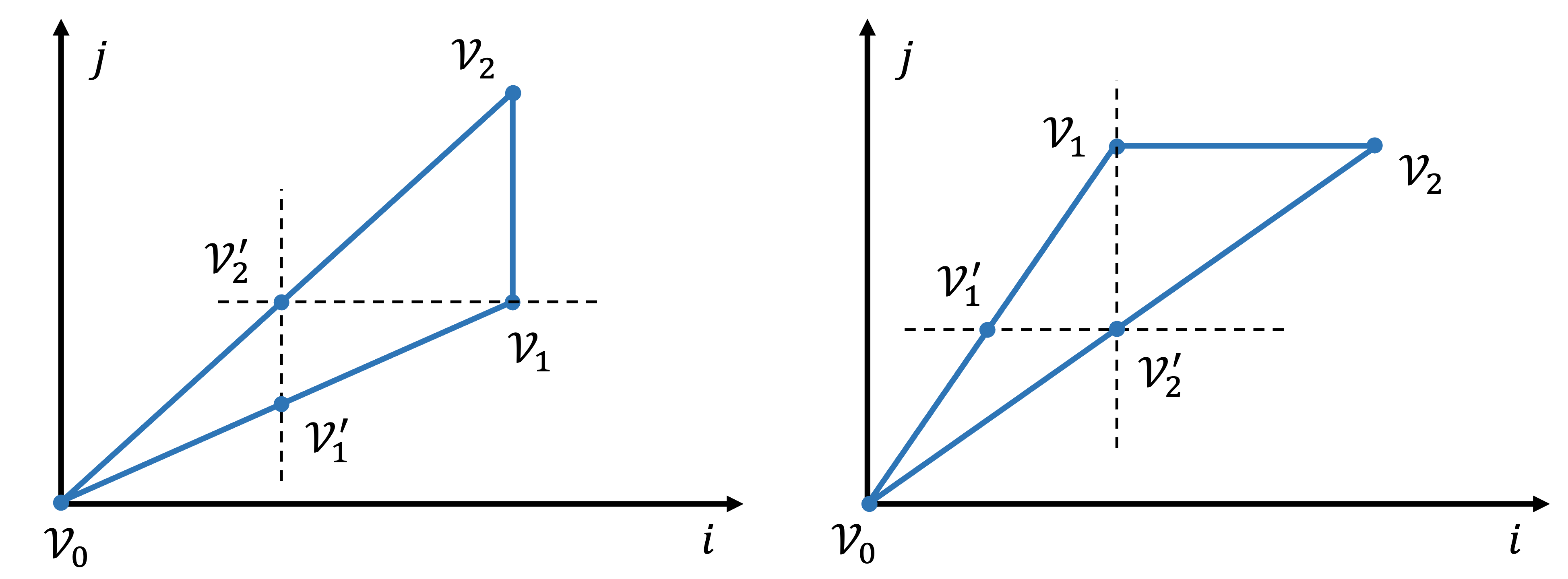}
\caption{Non-covered corner vertex: $\mathcal{V}_{0}$ lands \textit{outside} of the other two vertices when projected onto the linear space of the vertical edge (left) or when projected onto the space of the invariant edge (right). Each case requires a single vertical and horizontal cut to produce two good right triangles in the base case and a third triangle, which is a homothetic scaling of the original triangle.}
\label{fig:fig:third-vertex-uncovered}
\end{figure*}

\begin{theorem}
\label{th:triangle}
The reduction over $[\mathcal{V}_{0}, \mathcal{V}_{1}, \mathcal{V}_{2}]$ can be simplified, i.e., computed with $O(1)$ reduction operations per element of the answer variable $Y$.
\end{theorem}

\begin{proof}
By Induction.
As the base case, consider $n\leq c$ for some constant $c$.  For such sufficiently small triangles, computing each element of $Y$ needs only $O(1)$ reduction operations.
In general, when $n$ is sufficiently large, let us assume by induction that $\Delta_1$
can be simplified to compute the section of the result until $\mathcal{V}_1'$ in $O(1)$ time per element of $Y$.  Another backward scan and a forward scan on the appropriate section of $X$ will yield the next section of $Y$ from $\mathcal{V}_1'$ to $\mathcal{V}_1$.
\end{proof}

The above argument can be carried over mutatis mutandis to the case where the non-residual edge is horizontal.  Indeed, as illustrated on the right side of Figure~\ref{fig:fig:third-vertex-uncovered}, the first cut of the previous case yields the second one.

The proof of Theorem~\ref{th:triangle} suggests a simple code generation strategy (again, explained for the case when the non-residual edge is vertical).  Note that vertices $\mathcal{V}_{1}$ and $\mathcal{V}_{2}$ are given as affine functions of the parameter, $n$, and this is known statically.  Let 
Let $\mathcal{V}_{1} = [an, bn]$ and $\mathcal{V}_{2}  = [an, b'n]$ for some scalars $a, b$ and $b'$.
We can compute the values of $a_1, b_1$ and $b'_1$ the factors that specify $\mathcal{V}_1'$ and $\mathcal{V}_2'$, and from that, the scaling factor, $\frac{a}{a_0}=\frac{b}{b'_0}=\frac{b'}{b'_0}$.
This leads to the code structure shown in Figure~\ref{fig:codegen}.
It is important to note that this code is recursive, not polyhedral, and takes us out of the polyhedral model.
Nevertheless, all our analyses are polyhedral, and in order to generate it, we need only a fixed number of splits.
Also, note that the recursive function calls can be optimized using standard tail recursion optimization techniques.

\begin{figure}[h]
\begin{lstlisting}[style=CStyle, frame=none]
void fractal(int *Y, int *X, int L, int U) {
    if (U < threshold) {                       // do the full input reduction
        for (i=L; i<U; i++)
            for (j=i; j<=2*j; j++)
                Y[i] = max(Y[i], X[j]);
        return;
    }
    for (i=U; i>=U/2; i--)                     // backward scan on U/2<=i<=U
        Y[i] = max(Y[i+1], X[i]);
    for (i=U/2; i<=U; i++)                     // foward scan on U/2<=i<=U
        Y[i] = max(Y[i-1], X[2*i], X[2*i-1]);
    fractal(Y, X, L, U/2);                     // recurse on L<=i<U/2
} 
\end{lstlisting} 
\vspace{-2mm}
\caption{Recursive pseudo-code for fractal simplification of the example from Section~\ref{sec:motiv-i-to-2i}. The complete working code can be found in the artifact~\cite{narmour_maximal_2024}. }
\label{fig:codegen}
\end{figure}
\subsection{Three Residual Edges}

In this case any vertical cut through a covered vertex when projected onto the $i$-axis will produce two triangles in case 2.  Alternatively, and equivalently, any horizontal cut through a covered vertex when projected on the $j$-axis will achieve the same thing.

\section{Simplification With Splitting} \label{sec:implementation}

In this section, we summarize the extended optimal simplification algorithm that incorporates our previously proposed splitting techniques.

\subsection{Maximal Simplification Algorithm}\label{sec:new-algorithm}

Given an input reduction in the form of Equation~\ref{eq:reduction-def}, the maximal optimal simplification algorithm is summarized below in Algorithm~1.
This should be understood precisely as the dynamic programming algorithm of GR06 (Algorithm 2 in their work) with an additional dynamic programming step to carry out splitting when necessary.
Additionally, we require that the input reduction be initially preprocessed and separated into the union of disjoint simplices.
We assume that has already been done or viewed as an additional preprocessing step as needed.

\SetKwComment{Comment}{/* }{ */}
\begin{algorithm} 
    \caption{Maximal Optimal Simplification}
    \KwIn{$d$-dimensional reduction with $a$-dimensional accumulation and $r$-dimensional reuse. Assume that the body is a simplex and that the other preprocessing steps of GR06 have been done to expose the $r$ dimensions of reuse.}
    \KwOut{Equivalent optimal $\big(d-\min(a,r)\big)$-dimensional reduction(s)}
    \nl \While{there are residual reductions with reuse}{
        \nl \ForEach{residual reduction with body $\mathcal{D}$}{
            \nl Construct the set of candidate single-step simplifications among facets of $\mathcal{D}$ \\
            \nl Construct the set of candidate reduction decompositions to transform one (or more) facets of $\mathcal{D}$ into strong boundaries \\
            \nl \If{there are no possible single-step simplification and decomposition candidates}{
               \nl Construct the set of candidate SPB and SPI splits \\
            }
            \nl Use dynamic programming to optimally choose: \\
            {
                \nl \hspace{4mm}(a) A single-step simplification of $\mathcal{D}$ along candidate $\halfvec{\rho}$ \\
                \nl \hspace{4mm}(b) A reduction decomposition candidate \\
                \nl \hspace{4mm}(c) A SPB or SPI candidate split to produce two new residual reductions
            }
        }
    }
    \label{alg:msr}
\end{algorithm}

As stated previously, arbitrarily splitting reductions can lead to an endless loop; however, we can guarantee termination by only making simplex preserving strong boundary (SPB) and invariant (SPI) splits.
Furthermore, we only do so when no other valid candidate choices exist.
This is because the dynamic programming algorithm enumerates a finite number of choices based on the number of facets at the target reductions. Such splits strictly reduce the number of choices that need to be explored.
Each candidate split in step 6 in Algorithm~1 consists either of a single split, for 3-dimensions and higher per Section~\ref{sec:dD-reductions}, or two splits that expose a repeating pattern, in 2-dimensions per Section~\ref{sec:2D-reductions (triangles)}.

\subsection{Implementation}

We provide a proof-of-concept implementation of the individual components of our approach using the Alpha language~\cite{mauras_alpha_1989, le_verge_alpha_1991} and the AlphaZ system~\cite{yuki_alphaz_2013}.
The Alpha language is an equational language that separates the specification of a program from its execution plan.
Additionally, it supports modeling reduction operations as first-class objects with explicit representations of the write and read functions, $f_{p}$ and $f_{d}$, characterizing the accumulation and reuse space, respectively.
Program variable domains are represented as polyhedral sets using \texttt{isl}~\cite{verdoolaege_isl_2010} (the integer set library), which naturally supports the constraint representations we use to describe our splitting hyperplanes.

Our results in this paper are primarily theoretical.  We provide, in our artifact~\cite{narmour_maximal_2024}, complete working code for the components of the maximal simplification algorithm, including constructing the set of all possible labelings and corresponding candidate splits.  A complete push-button tool that implements the extended dynamic programming algorithm and automatically and optimally simplifies any reduction requires addressing several practical issues. Specifically, it is necessary to handle two issues: managing the combinatorially large number of possible solutions and developing methods to compare them using the constant factors, and the existence of parallel schedules with scalable locality. Putting all this together is outside the scope of this paper and left as future work.

This can be used to produce simplified Alpha programs for all of the examples discussed above.
Additionally, AlphaZ can be used to generate C code that can be subsequently compiled and run.
More information can be found in the accompanying artifact~\cite{narmour_maximal_2024}.

\subsection{A Complete Higher Dimensional Example}

Consider the $O(N^2)$ simplification of the following $O(N^4)$ reduction over the simplicial domain $\mathcal{D} = \{ [i,j,k,l] \; | \; j \leq N  \; \mathrm{and} \;  i \leq k \leq 2i \; \mathrm{and} \; i+j \leq l \leq 2j \}$ with 2-dimensional accumulation and reuse,
\begin{equation} \label{eq:ex-3-alpha}
Y[i,j] = \max_{[i,j,k,l] \in \mathcal{D}} X[k,l] 
\end{equation}
This is an interesting example because our splitting extensions are necessary, and among all possible simplification choices, with splitting, the fractal simplification step of Section~\ref{sec:2D-reductions (triangles)} is unavoidable.
Since 4-dimensional spaces are difficult to visualize, we illustrate each step by showing concrete loops.
Due to space constraints, we only describe the details relevant to the key effect of each step.
Complete working code can be found with the accompanying artifact~\cite{narmour_maximal_2024}.

\begin{figure*}[tb]
  \centering
\includegraphics[width=1\linewidth]{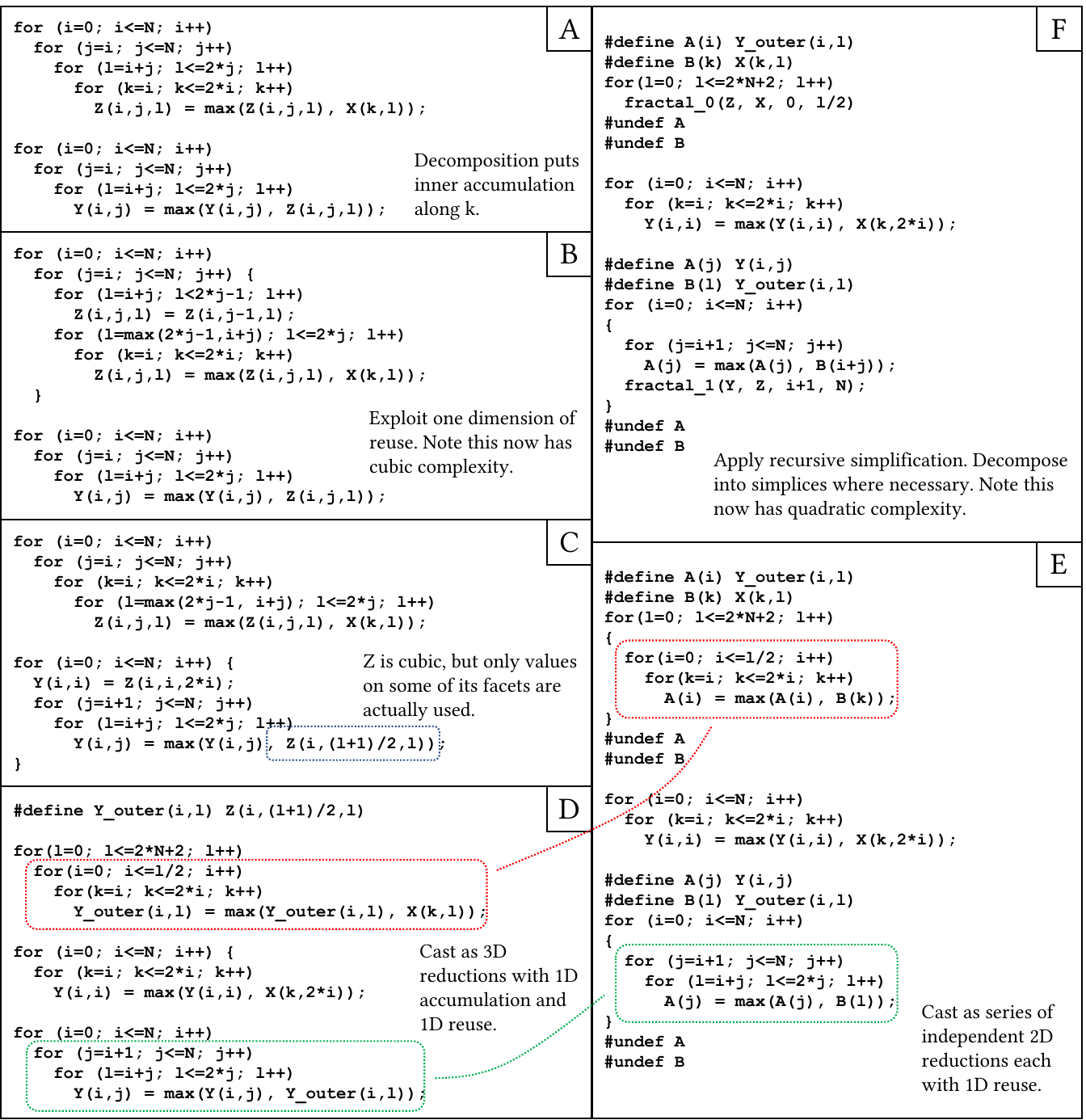}
\caption{Illustrating the $O(N^2)$ simplification procedure of the $O(N^4)$ reduction in Equation~\ref{eq:ex-3-alpha}.}
\label{fig:ex3-steps}
\end{figure*}

\textit{\textbf{Step A} - reduction decomposition to put the inner accumulation along k.}
The set of candidate single-step simplifications (step 3 of Algorithm~1) is empty.
The only option is to explore reduction decomposition choices.
Choosing a decomposition that puts the inner reduction along $k$ and storing the partial answer in a temporary intermediate variable $Z$ results in the loops shown in box A of Figure~\ref{fig:ex3-steps}.

\textit{\textbf{Step B} - simplification of the inner reduction.}
The set of candidate single-step simplifications on the inner reduction is now non-empty.
Applying single-step simplification using the reuse vector $\halfvec{\rho} = [0,1,0,0]$ results in the $O(N^3)$ loops shown in box B.

\textit{\textbf{Step C} - recover the residual reductions.}
The single-step simplification by $\halfvec{\rho} = [0,1,0,0]$ in this instance results in computing a 3-dimensional intermediate answer but GR06 describes and handles the post-processing required to read only the needed values in the outer reduction.
$Z$ is cubic, but only values on some of its facets are actually used.

\textit{\textbf{Step D} - preprocess to cast as series of independent reductions.}
The residual reductions now have 3-dimensional bodies with 1-dimensional accumulation and reuse, which corresponds to the fact that we are left with a series of 2D reductions embedded in a 3D space, as described in Section~\ref{sec:assumption-canonic-axes}.

\textit{\textbf{Step E} - cast as series of independent reductions.}
There are two residual 2D reductions embedded in a 3D space.
Both of these involve triangular domains that require the fractal simplification of Section~\ref{sec:2D-reductions (triangles)}.

\textit{\textbf{Step F} - Apply recursive fractal simplification.}
Finally, computing the inner 2D reductions recursively results in the loops shown in box F.
We don't explicitly show the definitions of \texttt{fractal\_0} and \texttt{fractal\_1} here, but they have the same structure as the code in Figure~\ref{fig:codegen}.
The complete concrete code for this example can be found in the accompanying artifact~\cite{narmour_maximal_2024}.

\section{Related Work} \label{sec:related-work}
Simplification has garnered renewed interest recently.  Asymptotic inefficiencies are present, even in deployed codes.  Ding and Shen~\cite{ding_glore_2017} noted that nine of the 30~benchmarks in Polybench~3.0 and two deployed PDE solvers have such inefficiencies. Separately, Yang et al.~\cite{yang_simplifying_2021} showed that simplification is helpful for many algorithms in statistical learning like Gibbs Sampling (GS), Metropolis Hasting (MH), and Likelihood Weighting (LW).  Their benchmarks include Gaussian Mixture Models (GMM), Latent Dirichlet Allocation (LDA), and Dirichlet Multinomial Mixtures (DMM).  See their paper for details of benchmarks, algorithms, size parameters, machine specs, etc.
They formulate and solve a more general problem: the reduction body may use (possibly transitively, through other variables) the output variable $Y$. So, it is necessary to solve the combined problem of simplification and scheduling.  They propose a simple heuristic that handles all the examples they encounter and only applies to reduction operators that admit an inverse.  Addressing these limitations, possibly using the techniques we present here, is an open problem.

Simplification is related but complementary to the problems of \emph{marginalization of product functions} (MPF)  and its discrete version,  \emph{tensor contraction} (TC).  Such problems arise in many domains.  MPF can be optimally computed using Pearl's ``summary passing'' or ``message passing'' algorithm~\cite{pearl_probabilistic_1988} for Bayesian inference, or the generalized distributive law~\cite{aji_generalized_2000, kschischang_factor_2001, shafer_probability_1990}.   Similarly, there is a long history of research on optimal implementations of TC~\cite{pfeifer_faster_2014,sabin_tensor_2021,kong_automatic_2023}.  In this problem, the sizes of the tensors in each dimension are known, and we seek the implementation with the minimum number of operations.

To explain the problem, first note that the cost of multiplying three matrices, $A$, $B$, and $C$, is affected by how associativity is exploited: if A and C are short-stout, and B is tall-thin, $(AB)C$ is better than $A(BC)$, and if A and C are tall-thin and B is short-stout, the latter parenthesization is.
Indeed, optimizing this for a sequence of matrices is a classic textbook problem used to illustrate dynamic programming. However, the underlying matrix multiplication uses the standard cubic algorithm. TC extends this to multiple chains of products (requiring us to optimally identify and exploit common subexpressions) and to tensors rather than matrices (exposing opportunities to exploit simplification by inserting new variables).

Specifically, consider the following system of equations:
\begin{align}
    X_{i,l} = \sum_{j,k=1}^{N}A_{i,j,l} \times B_{i,k}
    &&
    Y_j = \sum_{i,k,j=1}^{N}A_{i,j,l} \times B_{i,k}
\end{align}
Naively, each equation would have $O(N^{4})$ complexity since each result is the accumulation of a triple summation, and there are  $O(N)$ answers. However, if we define two new variables:
\begin{align}
    T_{i,l}  = \sum_{j=1}^{N}A_{i,j,l}
    &&
    T'_{i} = \sum_{k=1}^{N}B_{i,k}
\end{align}
then the following is an equivalent simplified system of equations whose complexity is only $O(N^{3})$:
\begin{align}
    X_{i,l} = T_{i,l} \times T'_{i}
    &&
    Y_j = \sum_{i,l=1}^{N}A_{i,j,l} \times T'_{i}
\end{align}


However, the reuse space, the accumulation space, and the domain constraints (loop bounds) are simple, and simple transformations like loop permutation can explore the space of all possible simplifications.
The space of choices may be combinatorially large but not infinite, as tackled by the GR06 algorithm.  There may be a need for more sophisticated) simplifications if the tensors have a special structure.
Nevertheless, the reduction operation here admits an inverse, and there is no need for our results in this paper.

\section{Conclusions and Open Questions} \label{sec:conclusion}



The ultimate goal is to enable compilers to take a high-level application program specification and carry it out in the most efficient way possible, preferably \textit{automatically} and \textit{optimally}.
This work takes a step in that direction to enable users (i.e., application scientists and programmers) to focus less on the \textit{engineering} aspects (the \textit{how}) of their algorithms and more on the problems (the \textit{what}) that their algorithms are intended to solve.
In this work, we have studied reuse-based simplification of polyhedral reductions.
The simplification transformation proceeds recursively down the face lattice of the domain of the reduction body and attempts to exploit one dimension of available reuse at each level.
Previously, at some point along this traversal, it may not have been possible to employ the simplification transformation without requiring inverse operations.
However, as we have shown, it is always possible to split the residual reduction at the problematic node in the lattice in such a way that we can guarantee that the simplification transformation will not fail.
We have provided the mathematical proofs which enable us to make this claim.

We showed how to maximally simplify any arbitrary independent commutative reduction to obtain the optimal asymptotic complexity.
We provided a proof-of-concept implementation of the individual components of our approach as an accompanying software artifact~\cite{narmour_maximal_2024} that can be used to generate code.
Along these same lines, additional work concerning traditional compile-time scheduling and efficient code generation is still needed. 
For example, we rely on the fact that any polyhedral set can be decomposed into a union of disjoint simplices, but we do not discuss the implications of the choice of decomposition on the efficiency of the resulting code.
At this point, for example, it is not obvious why we should prefer one decomposition over another.
It may be the case that one particular simplicial decomposition scheme leads to code that is more amenable to vectorization.
These are interesting questions that require future exploration.

\begin{acks}
    The work of the first author was partially supported by the \grantsponsor{GS1}{National Science Foundation}{http://dx.doi.org/10.13039/100000001} under Grant No.~\grantnum{GS1}{2318970}.
\end{acks}

\bibliographystyle{ACM-Reference-Format}
\bibliography{references}

\end{document}